\documentclass[journal]{IEEEtran}
\usepackage[normalem]{ulem}
\usepackage{color}
\usepackage[noadjust]{cite}
\usepackage{amsmath}
\usepackage{algorithmic}
\usepackage{algorithm}
\usepackage{amsthm}
\usepackage{graphicx}
\usepackage{amssymb}
\usepackage{float}
\renewcommand\qedsymbol{$\blacksquare$}
\DeclareMathOperator{\E}{\mathcal{E}}
\DeclareMathOperator*{\argmax}{arg\,max}
\DeclareMathOperator*{\argmin}{arg\,min}
\newtheorem{theorem}{Theorem}

\newtheorem{lemma}[theorem]{Lemma}
\usepackage{xcolor}
\begin{document}
\title{Secrecy Analyses of a Full-Duplex MIMOME Network}
\author
    {Reza Sohrabi,~\IEEEmembership{Student Member,~IEEE,}
    Qiping Zhu,~\IEEEmembership{Student Member,~IEEE,}
        Yingbo Hua,~\IEEEmembership{Fellow,~IEEE}\\
        \thanks{The authors are with
Department of Electrical and Computer Engineering,
University of California, Riverside, CA 92521, USA. Emails: rsohr001@ucr.edu,  qzhu005@ucr.edu and yhua@ece.ucr.edu. R. Sohrabi and Q. Zhu have both contributed to this work as the first author.
This work was supported in part by the Army Research Office under Grant Number W911NF-17-1-0581. The views and conclusions contained in this
document are those of the author and should not be interpreted as representing the official policies, either
expressed or implied, of the Army Research Office or the U.S. Government. The U.S. Government is
authorized to reproduce and distribute reprints for Government purposes notwithstanding any copyright
notation herein.
}}

\markboth{ }
{Shell \MakeLowercase{\textit{et al.}}: Bare Demo of IEEEtran.cls for Journals}

\maketitle

\begin{abstract}
    This paper presents secrecy analyses of a full-duplex MIMOME network which consists of two full-duplex multi-antenna users (Alice and Bob) and an arbitrarily located multi-antenna eavesdropper (Eve). The paper assumes that Eve's channel state information (CSI) is completely unknown to Alice and Bob except for a small radius of secured zone. The first part of this paper aims to optimize the powers of jamming noises from both users. To handle Eve's CSI being unknown to users, the focus is placed on Eve at the most harmful location, and the large matrix theory is applied to yield a hardened secrecy rate to work on. The performance gain of the power optimization in terms of maximum tolerable number of antennas on Eve is shown to be significant. The second part of this paper shows two analyses of anti-eavesdropping channel estimation (ANECE) that can better handle Eve with any number of antennas. One analysis assumes that Eve has a prior statistical knowledge of its CSI, which yields lower and upper bounds on secure degrees of freedom of the system  as functions of the number (N) of antennas on Eve and the size (K) of information packet. The second analysis assumes that Eve does not have any prior knowledge of its CSI but performs blind detection of information, which yields an approximate secrecy rate for the case of K being larger than N.
\end{abstract}

\section{Introduction}
Security of wireless networks is of paramount importance in today's world as  billions of people around the globe are dependent upon these networks for a myriad of activities for their businesses and lives. Among several key issues in wireless security  \cite{Zou-Zhu2016}, confidentiality is of particular interest to many researchers in recent years and is a focus of this paper.  For convenience, we will refer to confidentiality as security and vice versa.

The traditional  way to keep information confidential from unauthorized persons and/or devices is via cryptography at upper layers of the network, which include the asymmetric-key method (involving a pair of public key and private key) and the symmetric-key method (involving a secret key shared between two legitimate users). As the computing capabilities of modern computers (including quantum computers) rapidly improve, the asymmetric-key method is increasingly vulnerable as this method relies on computational complexity for security. In fact, the symmetric-key method is gaining more attraction in applications \cite{Koziol2018}.

However, the establishment of a secret key (or any secret) shared between two users is not trivial in itself. Even if a secret key was pre-installed in a pair of legitimate devices (during manufacturing or otherwise), the lifetime of the secret key in general  shortens each time the secret key is used for encryption. For many applications such as big data streaming, such secret key must be periodically renewed or changed. To enjoy the convenience of mobility, it is highly desirable for users to be able to establish a secret key in a wireless fashion.

Establishing a secret key or directly transmitting secret information  between users in a wireless fashion (without a pre-existing shared secret)  is the essence of physical layer security \cite{Bloch2011}. There are two complementary approaches in physical layer security: secret-key generation and secret information transmission. The former requires users to use their (correlated) observations and an unlimited public channel to establish a secret key, and the latter requires one user to transmit secret information directly to the other. This paper is concerned with the latter, i.e., transmission of secret information (such as secret key) between users without any prior digital secret.

Specifically, this paper is focused on a network as illustrated in Fig. \ref{fig:config1} where one legitimate user (Alice) wants to send a secret key to another legitimate user (Bob) subject to eavesdropping by an eavesdropper (Eve) anywhere. Each of the two users/devices is allowed to have multiple antennas, and both Alice and Bob are capable of full-duplex operations. Following a similar naming in the literature such as \cite{5605343}, we call the above setup a full-duplex MIMOME network where MIMOME refers to the multi-input multi-output (MIMO) channel between Alice and Bob and the multi-antenna Eve.

The MIMOME related works in the literature include: \cite{5961840, khisti2007gaussian, 4626059, 6848758, 4529293, liu2009note, 5605343, 5940242, 7219473, zhou2010secure, cepheli2017joint, masood2017minorization} where the channel state information (CSI) at Eve is assumed to be known not only to Eve itself\footnote{All entities are treated as ``gender neutral''.} but also to Alice and Bob; \cite{1558439, 5550916, 5699835, 7328736} where a partial knowledge of Eve's CSI is assumed to be available to Alice and Bob and an averaged secrecy or secrecy outage was considered;  \cite{negi2005secret,4543070,liao2011qos, zhou2010secure, swindlehurst2009fixed, 5161804, liu2015artificial} where  artificial noise is embedded in the signal from Alice; and \cite{zheng2013improving,zhou2014application,7339654,hua2017fundamental, chen2017fast, hua2018advanced} where Bob is treated as a full-duplex node capable to receive the signal from Alice while transmitting jamming noise.

From the literature, the idea of using jamming noise from Alice or Bob appears important.
Inspired by that, this paper will first consider a case where both Alice and Bob send jamming noises while Alice transmits secret information to Bob. We will explore how to optimize the jamming powers from Alice and Bob. In \cite{zhou2014application}, jamming from both users was also considered. But here for power optimization we include the effect of the residual self-interference of full-duplex radio. There are other differences in the problem formulation and objectives.   We assume that Eve's CSI is completely unknown to Alice or Bob except for a radius of secured zone free of Eve around Alice. A similar idea was also applied in \cite{liu2016ergodic} but in a different problem setting. We will focus on Eve that is located at the most harmful position. Furthermore, to handle the small-scale fading at Eve, we apply the large matrix theory to obtain a closed-form expression of a secrecy rate, which makes the power optimization tractable. Unlike \cite{liu2015artificial} where large matrix theory was also applied, we consider an arbitrary large-scale-fading at Eve among other major differences.
With the optimized powers, we reveal a significant performance gain in terms of the maximum tolerable number of antennas on Eve to maintain a positive secrecy. We will also show that as the number of antennas on Eve increases, the impact of the jamming noise from either Alice or Bob on secrecy vanishes. This contribution extends a previous understanding of single-antenna users shown in \cite{hua2018advanced}.

Later in this paper, we will analyze a two-phase scheme for secret information transmission proposed in \cite{hua2018advanced}. In the first phase, an anti-eavesdropping channel estimation (ANECE) method is applied which allows users to find their CSI but suppresses Eve's ability to obtain its CSI. In the second phase, secret information is transmitted between Alice and Bob while Eve has little or no knowledge of its CSI. We show two analyses based on two different assumptions. The first analysis assumes that Eve has a prior statistical knowledge of its CSI. With every node knowing a statistical model of CSI anywhere, we use mutual information to analyze the secret rate of the network, from which lower and upper bounds on the secure degrees of freedom are derived. These bounds are simple functions of the number of antennas on Eve. The second analysis assumes that Eve does not have any prior knowledge of its CSI. Due to ANECE in phase 1, Eve is blind to its CSI. But in phase 2, Eve performs blind detection of the information from Alice. We analyze the performance of the blind detection, from which an approximate secret rate is derived and numerically illustrated. Both of these analyses are important contributions useful for a better understanding of ANECE.

\textit{Notation}: Matrices and column vectors are denoted by upper and lowercase boldface letters. The trace, Hermitian transpose, column-wise vectorization, $(i,j)$th element, and complex conjugate of a matrix $\mathbf{A}$ are denoted by $\mathrm{Tr}\left (\mathbf{A}\right)$, $\mathbf{A}^H$, $\textrm{vec}\left (\mathbf{A}\right)$, $\mathbf{A}_{i,j}$, and $\mathbf{A}^*$, respectively.
For a matrix $\mathbf{X}$ and its vectorized version $\mathbf{x}$, $\textrm{ivec}\left (\mathbf{x}\right)$ is the inverse operation of $\mathbf{x}=\textrm{vec}\left (\mathbf{X}\right)$. A diagonal matrix with elements of $\mathbf{x}$ on its diagonal is $\mathrm{diag}\left (\mathbf{x}^T\right)$. Expectation with respect to a random variable $x$ is denoted by $\mathcal{E}_x\left [\cdot\right ]$. Let the random
variables ${X_n}$ and $X$ be defined on the same probability space, and we write $X_n \overset{a.s.}{\rightarrow}X$ if $X_n$ converges to $X$ almost surely as $n \rightarrow \infty$. The identity matrix of the size $n\times n$ is $\mathbf{I}_n$ (or $\mathbf{I}$ with $n$ implied in the context), and $\mathbf{1}_n $ is a row vector of length $n$ of all ones. A circularly
symmetric complex Gaussian random variable $x$ with variance $\sigma^2$ is denoted as $x  \sim \mathcal{CN}\left (0, \sigma^2\right)$. The mutual information between random variables $x$ and $y$ is $I\left (x;y\right)$, and  $h(x)$ denotes the differential entropy of $x$. Logarithm in base 2 is denoted by $\log\left (\cdot\right)$,
and $\left (\cdot\right)^+ \triangleq \max\left (0,\cdot\right)$.

\section{Optimization of Jamming Powers and Effects of Eve's Antennas}\label{sec:opt}
\subsection{System Model}
\label{sec:model}
Our network setup is shown in Fig. \ref{fig:config1},
\begin{figure}[bt!]
    \centering
    \includegraphics[width=.8\linewidth]{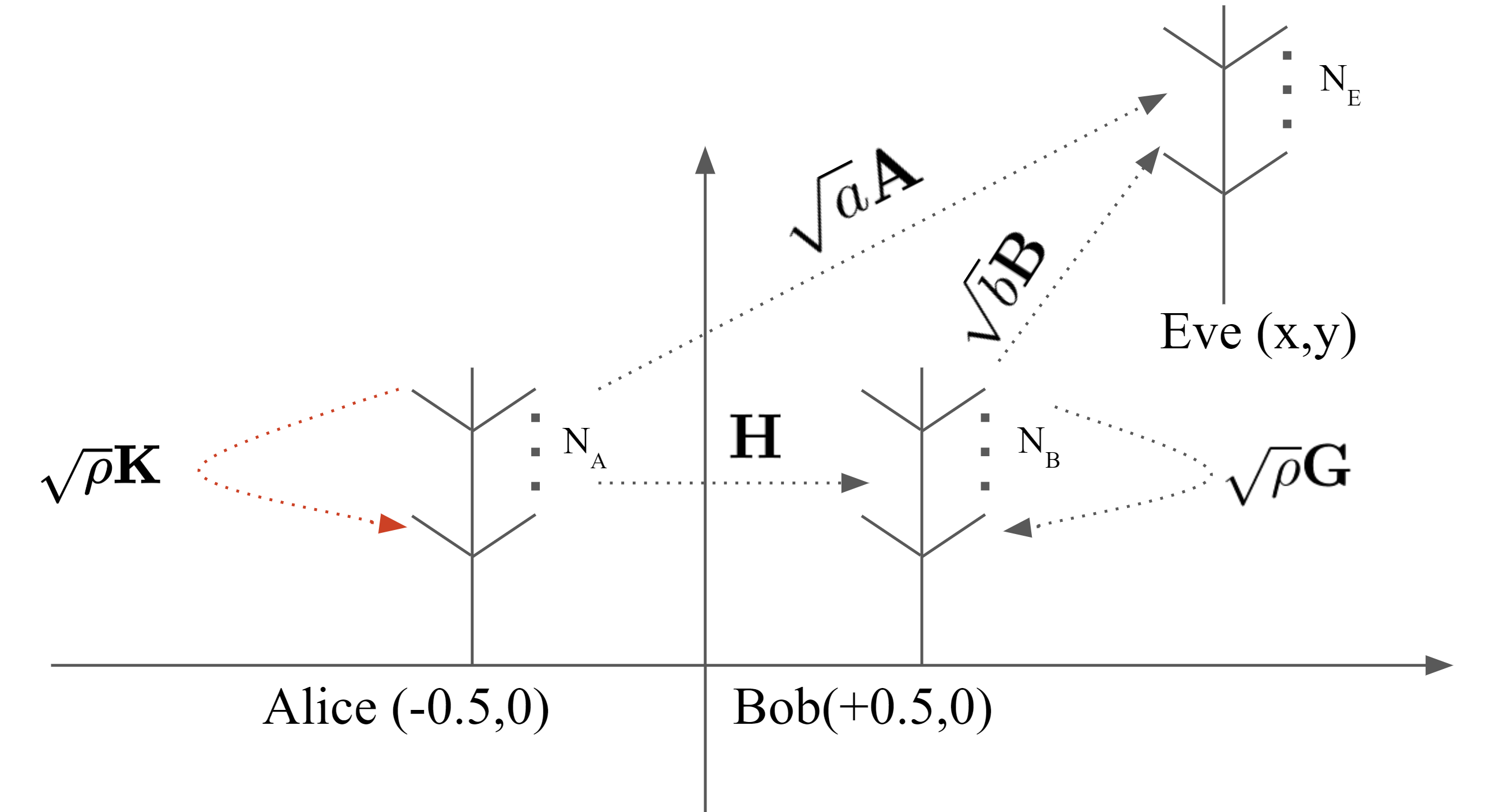}
    \caption{A Full-Duplex MIMOME network}
    \label{fig:config1}
\end{figure}
where Alice   (with $N_A$ antennas) intends to send secret information over a wireless channel to Bob   (with $N_B$ antennas) in the presence of
 possibly many passive Eves   (of $N_E$ antennas each) that may collude with each other at the network layer but not at the physical layer. We will focus on the most harmful Eve.
 Physical layer colluding among distributed Eves to form a large virtual antenna array is highly difficult in practice. But if a virtual antenna array from colluding Eves is likely in some applications, we could treat these colluding Eves as a single mega Eve with a large number of antennas.

The system parameters are normalized in a similar way as in \cite{hua2017fundamental}. In particular, the large-scale-fading factor from Alice to Eve is modeled as (when a model is needed):
$
    a={d_A^{-\alpha}}={\left (\left (x+0.5\right)^2+y^2\right)^{-\alpha/2}},
$
and that from Bob to Eve is
$
    b= {d_B^{-\alpha}}={\left (\left (x-0.5\right)^2+y^2\right)^{-\alpha/2}}
$
where $\alpha$ is the path-loss exponent. We assume that no Eve is closer to Alice than a radius $\Delta$, i.e., $d_A \geq \Delta$. The normalized large-scale-fading factor of the residual self-interference at both Alice and Bob is denoted by $\rho$. (In all simulations, $\rho$ is considered to be $0.1\%$.) The small-scale-fading channel matrix from Alice to Eve is denoted by $\mathbf{A}$, that from Bob to Eve is $\mathbf{B}$, and that of the residual self-interference at Bob and Alice are $\mathbf{G}$ and $\mathbf{K}$, respectively\footnote{Up to Section \ref{sec:anece}, Alice is only a transmitter, and hence it does not utilize its full-duplex capability.}. The channel matrix  from Alice to Bob is denoted by $\mathbf{H}$, and its SVD is denoted by
\begin{equation}\label{eq:H}
    \mathbf{H}= \mathbf{U}\boldsymbol{\Sigma} \mathbf{V}^H,
\end{equation}
where $\mathbf{U}$ and $\mathbf{V}$ are unitary matrices, and $\boldsymbol{\Sigma}$ is the $N_B\times N_A$ diagonal matrix that contains the singular values of $\mathbf{H}$ (i.e., $\sigma_i$, $i=1,\cdots,N_B$) in descending order  assuming $N_A\geq N_B$.
All the elements in all channel matrices are modeled as i.i.d. circularly symmetric complex Gaussian random variables with zero mean and unit variance.

In this section, we assume that Alice and Bob have the knowledge of $\mathbf{H}$ but not of $\mathbf{A}$ and $\mathbf{B}$, and Eve has the knowledge of all these matrices.

Alice sends the following signal containing $r\leq N_B\leq N_A$ streams of secret information mixed with artificial noise:
\begin{equation}
\label{eq:x_A}
    \mathbf{x}_A(k)= \mathbf{V}_1\mathbf{s}(k)+\mathbf{V}_2 \mathbf{w}_A(k),
\end{equation}
where $k$ is the index of time slot, $\mathbf{V}_1$ is the first $r$ columns of $\mathbf{V}$, $\mathbf{V}_2$ is  the last $N_A-r$ columns of $\mathbf{V}$, $\mathbf{s}(k)$ is Alice's information vector with the covariance matrix $\mathbf{Q}_r$ and $\mathrm{Tr}\left (\mathbf{Q}_r\right)= P_s$, and $\mathbf{w}_A(k)$ is an $\left (N_A-r\right)\times 1$ artificial noise vector with distribution $\mathcal{CN}\left (\mathbf{0},{\frac{P_n}{N_A-r}}\mathbf{I}\right)$. Here, $P_s+P_n = P_A \leq P_A^{max}$.

While Bob receives information from Alice, it also sends a jamming noise:
\begin{equation}
\label{eq:x_B}
    \mathbf{x}_B(k)=  \mathbf{w}_B(k),
\end{equation}
where $ \mathbf{w}_B(k)$ is an $N_B\times 1$ artificial noise vector with distribution $\mathcal{CN}\left (\mathbf{0},{\frac{{P}_B}{N_B}}\mathbf{I}\right)$.

Note that both ${P}_A$ and $ {P}_B$ are  normalized powers with respect to the path loss from Alice to Bob, and with respect to the power of the background noise. So, without loss of generality, we let the power of the background noise be one.

With jamming from both Alice and Bob, the signals received by Bob and Eve are respectively:
\begin{equation}
\label{eq:yb}
    \mathbf{y}_B(k)= \mathbf{H} \mathbf{V}_1\mathbf{s}(k)+ \mathbf{H} \mathbf{V}_2 \mathbf{w}_A(k) +\sqrt{\rho} \mathbf{G} \bar{\mathbf{w}}_B(k)+{\mathbf{n}}_B(k),
\end{equation}
\begin{equation}
\begin{split}
\label{eq:ye}
    \mathbf{y}_E(k)=& \sqrt{a}\mathbf{A}_1\mathbf{s}(k)+\sqrt{a}\mathbf{A}_2 \mathbf{w}_A(k)+ \sqrt{b} \mathbf{B} \mathbf{w}_B(k)+\mathbf{n}_E(k),
\end{split}
\end{equation}
where $\left [\mathbf{A}_1, \mathbf{A}_2\right ]=\left [\mathbf{A}\mathbf{V}_1, \mathbf{A}\mathbf{V}_2\right ]=\mathbf{AV}$. Since $\mathbf{A}_1$ and $\mathbf{A}_2$ are linear functions of the Gaussian matrix $\mathbf{A}$, they remain Gaussian. Because of the unitary nature of $\mathbf{V}$, $\mathbf{A}_1$ and $\mathbf{A}_2$ are independent of each other, and all elements in them are i.i.d. Gaussian of zero mean and unit variance. The noise vectors $\mathbf{n}_B$ and $\mathbf{n}_E$ are distributed as $\mathcal{CN}\left (\mathbf{0},\mathbf{I}\right)$. Also note that $\sqrt{\rho} \mathbf{G} \bar{\mathbf{w}}_B(k)$ is the residual self-interference originally caused by ${\mathbf{w}}_B(k)$ but is independent of ${\mathbf{w}}_B(k)$ \cite{hua2018advanced}.

If CSI anywhere is known everywhere,  the achievable secrecy rate of the above system is  known \cite{Csiszar1978} to be
\begin{equation}\label{}
  R_S=(R_{AB}-R_{AE})^+
\end{equation}
where $R_{AB}$ is the rate from Alice to Bob and $R_{AE}$ is the rate from Alice to Eve. Namely,
\begin{align}
\label{eq:R_AB}
 R_{AB} &= \log|\mathbf{I}+\mathbf{C}_B^{-1}\mathbf{H}\mathbf{V}_1\mathbf{Q}_r\mathbf{V}_1^H\mathbf{H}^H|,
\end{align}
\begin{align}
\label{eq:R_ae_orig}
     R_{AE} &= \log |\mathbf{I}+a \mathbf{C}_E^{-1}\mathbf{A}_1\mathbf{Q}_r\mathbf{A}_1^H|,
\end{align}
where
\begin{align}
\label{eq:R_B2}
    \mathbf{C}_B=\mathbf{I}+\frac{P_n}{N_A-r} \mathbf{H}\mathbf{V}_2\mathbf{V}_2^H\mathbf{H}^H+ \frac{\rho {P}_B}{N_B} \mathbf{G}\mathbf{G}^H,
\end{align}
\begin{align}
    \mathbf{C}_E=\mathbf{I}+{\frac{aP_n}{N_A-r}}\mathbf{A}_2\mathbf{A}_2^H+ {\frac{b {P}_B}{N_B}} \mathbf{B}\mathbf{B}^H.
\end{align}
Note that since $\mathbf{H}\mathbf{V}_1=\mathbf{U}_1\boldsymbol{\Sigma}_1$ and $\mathbf{H}\mathbf{V}_2=\mathbf{U}_2\boldsymbol{\Sigma}_2$ are orthogonal to each other where $\mathbf{U}_1$ and $\mathbf{U}_2$ are the partitions of $\mathbf{U}$ similar to those of $\mathbf{V}$, and $\boldsymbol{\Sigma}_1$ and $\boldsymbol{\Sigma}_2$ are the corresponding diagonal partitions of $\boldsymbol{\Sigma}$, a sufficient statistics of $\mathbf{s}(k)$ at Bob is $\mathbf{U}_1^H\mathbf{y}_B(k)=\boldsymbol{\Sigma}_1\mathbf{s}(k)+
\sqrt{\rho}\mathbf{U}_1\mathbf{G}\mathbf{\bar w}_B(k)+\mathbf{U}_1\mathbf{n}_B(k)$ which shows that the artificial noise from Alice does not affect Bob. Consequently, an equivalent form of $R_{AB}$ is
\begin{align}
\label{eq:R_AB_2}
 R_{AB} &= \log|\mathbf{I}_r+\mathbf{C}_{B,1}^{-1}\boldsymbol{\Sigma}_1\mathbf{Q}_r\boldsymbol{\Sigma}_1|,
\end{align}
where
\begin{align}
\label{eq:R_B2_2}
    \mathbf{C}_{B,1}=\mathbf{I}_r+ \frac{\rho {P}_B}{N_B} \mathbf{U}_1^H\mathbf{G}\mathbf{G}^H\mathbf{U}_1.
\end{align}
However, the expression shown in \eqref{eq:R_AB} is needed later due to its direct connection to $\mathbf{H}$ and $\mathbf{G}$.

If we optimize $P_n$, $P_B$ and $\mathbf{Q}_r$ to maximize the above $R_S$, the solution would be a function of Eve's CSI. This does not appear to be useful in practice.

If Eve's CSI is unknown to Alice but the statistics of Eve's CSI is known to Alice, then we can consider the ergodic secrecy:
\begin{equation}\label{}
  \bar R_S=(\mathcal{E}_{\mathbf{H},\mathbf{G}}[R_{AB}]-\mathcal{E}_{\mathbf{A},\mathbf{B}}[R_{AE}])^+
\end{equation}
which is achievable via coding over many CSI coherence periods. Closed form expression of each of the two terms in the above can be obtained using ideas in \cite{chiani2010mimo} and \cite{1203954}. But if we use $\bar R_S$ as objective to optimize $P_n$, $P_B$ and $\mathbf{Q}_r$, the solution would be independent of the CSI between Alice and Bob, and such a solution is not very useful either.

Because of the above reasons, we will consider the worst case of $R_{AE}$. The worst case is such that Eve is located at the most harmful location and has a large number of antennas.

It is shown in our earlier work \cite{sohrabi_secrecy} that the most harmful position of Eve is at $x^*=-0.5-\Delta$ and $y^*=0$. From now on, we will refer to $a$ and $b$ as corresponding to the position $\left (x^*,y^*\right)$. In all simulations, we will use $\Delta=0.1$ unless mentioned otherwise.

Given a large number of antennas at Eve, we can use large matrix theory to obtain a closed-form expression of  $R_{AE}$ that is no longer dependent on instantaneous CSI at Eve, which is shown next.
We can rewrite \eqref{eq:R_ae_orig} as follows:
  \begin{equation}\label{eq:R_ae_orig_b}
    R_{AE} =
    \log |\mathbf{I}+ \mathbf{J}_3\bar{\mathbf{\Theta}}_3\mathbf{J}_3^H|-\log |\mathbf{I}+ \mathbf{J}_4\bar{\mathbf{\Theta}}_4\mathbf{J}_4^H|
\end{equation}
where $\mathbf{J}_3=\frac{1}{\sqrt{N_E}}[\mathbf{A}_1,\mathbf{A}_2,\mathbf{B}]$, $\mathbf{J}_4=\frac{1}{\sqrt{N_E}}[\mathbf{A}_2,\mathbf{B}]$, and
\begin{equation}
    \bar{\mathbf{\Theta}}_3 = N_E\mathrm{diag}\left [ a\mathbf{q}_r^T,\frac{a{P}_n}{N_A-r}\mathbf{1}_{N_A-r},\frac{b P_B}{N_B} \mathbf{1}_{N_B}\right ],
\end{equation}
\begin{equation}
    \bar{\mathbf{\Theta}}_4 = N_E\mathrm{diag}\left [ \frac{a{P}_n}{N_A-r}\mathbf{1}_{N_A-r},\frac{b P_B}{N_B} \mathbf{1}_{N_B}\right ]
\end{equation}
where $\mathbf{q}_r$ is the vector containing the diagonal elements of the diagonal matrix $\mathbf{Q}_r$ (assuming that Alice does not know Bob's self-interference channel).
Note that the $\mathbf{J}$ matrices consist of i.i.d. random variables and the $\bar{\boldsymbol{\Theta}}$ matrices are diagonal. (The numbering of 3 and 4 used here is because of the numbering later.)

\begin{lemma}
\label{sec:lemma1}
Let $\mathbf{J}$ be an $N\times K$ matrix whose entries are i.i.d. complex random variables with variance $\frac{1}{N}$, and $\mathbf{\Theta}$ be a diagonal deterministic matrix. Based on Theorem   (2.39) of \cite{tulino2004random}, as $N, K\rightarrow \infty$ with $\frac{K}{N}\rightarrow\beta$, we have
\begin{equation}
\label{eq:asymptotic_rate}
    \frac{1}{N}\log |\mathbf{I}+ \mathbf{J}\mathbf{\Theta}\mathbf{J}^H|\overset{a.s.}{\rightarrow}\Omega\left (\beta, \mathbf{\Theta} ,\eta\right),
\end{equation}
where
\begin{align}
\label{eq:omega}
    \Omega\left (\beta, \mathbf{\Theta}, \eta\right) \triangleq \beta  \mathcal {V}_{\mathbf{\Theta} }\left (\eta \right) - \log \left (\eta\right) +\left (\eta -1\right)\log\left (e\right),
\end{align}
\begin{align}
\label{eq:shannon}
    \mathcal {V}_{\mathbf{\Theta} }\left (\eta \right)\triangleq \frac{1}{L_{\mathbf{\Theta}}}
    \sum_{j=1}^{L_{\mathbf{\Theta}}} \log\left (1+\eta\mathbf{\Theta}_{j,j}\right).
\end{align}
Here, $\mathbf{\Theta}_{j,j}$ is the $j$th diagonal element of the diagonal matrix $\mathbf{\Theta}$, $L_{\mathbf{\Theta}}$ is the number of diagonal elements of ${\mathbf{\Theta}}$, and $\eta>0$ is the solution to the equation
\begin{align}
\label{eq:eq_eta}
     {1-\eta}=\frac{\beta\eta}{L_{\mathbf{\Theta}}}
    \sum_{j=1}^{L_{\mathbf{\Theta}}} \frac{\mathbf{\Theta}_{j,j}}{1+\eta\mathbf{\Theta}_{j,j}}.
\end{align}
\end{lemma}
\begin{proof}
The proof is given in Appendix \ref{sec:append1}.
\end{proof}

Using the lemma, it follows from \eqref{eq:R_ae_orig_b} that for large $N_E$,
\begin{eqnarray}\label{eq:RAE_b}
  &&R_{AE} \nonumber\\
  &\simeq&\sum_{j=1}^{L_{\bar{\mathbf{\Theta}}_3}} \log\left (1+\bar{\eta}_3\left (\bar{\mathbf{\Theta}}_3\right)_{j,j}\right) -\sum_{j=1}^{L_{\bar{\mathbf{\Theta}}_4}} \log\left (1+\bar{\eta}_4\left (\bar{\mathbf{\Theta}}_4\right)_{j,j}\right)\nonumber\\
  &&+ N_E\log \left (\frac{\bar{\eta}_4}{\bar{\eta}_3}\right) +N_E\left (\bar{\eta}_3 -\bar{\eta}_4\right)\log\left (e\right)\nonumber\\
  &\triangleq& \mathcal{R}_{AE}
\end{eqnarray}
where for $i=3,4$, $\bar \eta_i$ is the solution of $\eta$ to \eqref{eq:eq_eta} with $\beta=\bar\beta_i$ and $\bar{\mathbf{\Theta}}=\bar{\mathbf{\Theta}}_i$. Here, $\bar\beta_3=\frac{N_A+N_B}{N_E}$ and $\bar \beta_4=\frac{N_A-r+N_B}{N_E}$. Note that the right side of \eqref{eq:eq_eta} is a monotonic function of $\eta\geq 0$ and hence a unique solution of $0\leq \eta \leq 1$ can be easily found by bisection search.

\begin{figure}[bt!]
    \centering
    \includegraphics[width=1\linewidth]{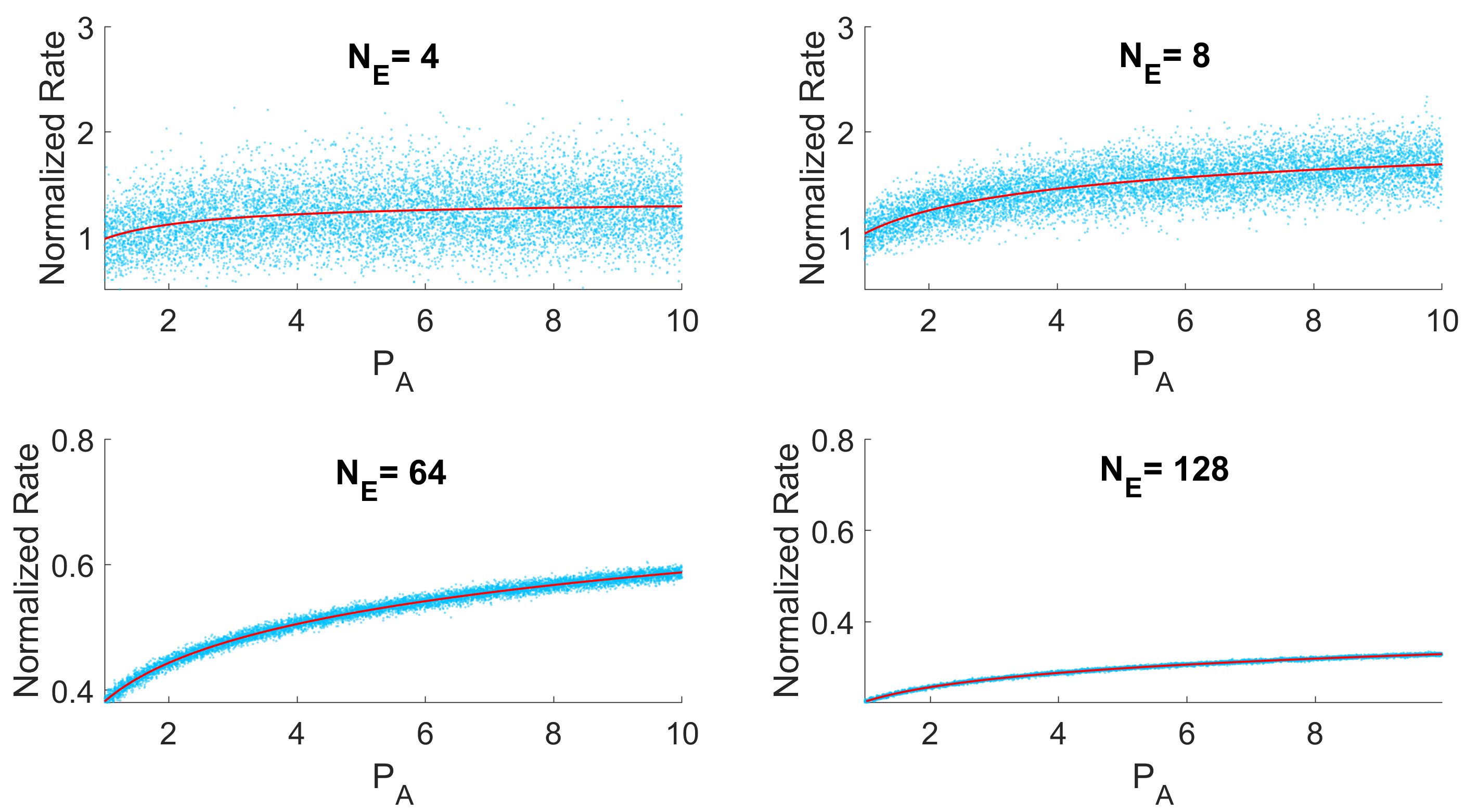}
    \caption{Comparison of the exact random realizations of $\frac{R_{AE}}{N_E}$ with its asymptotic result (the red/solid curve).}
    \label{fig:comp_asympt2}
\end{figure}

It is useful to note that the asymptotic form $\mathcal{R}_{AE}$ is a good approximation of the exact form $R_{AE}$ as long as $N_E$ is large regardless of $N_A$ and $N_B$.
Shown in Fig. \ref{fig:comp_asympt2} is a comparison of the exact random realizations of $\frac{R_{AE}}{N_E}$ from \eqref{eq:R_ae_orig} with its asymptotic result $\frac{\mathcal{R}_{AE}}{N_E}$ from \eqref{eq:RAE_b} where $N_A=2N_B=8$, $r=N_B$, $\mathbf{Q}_r=\frac{P_s}{N_B}\mathbf{I}$, $P_s=P_n=\frac{P_A}{2}$. Note that 100 realizations of $\frac{R_{AE}}{N_E}$ corresponding to 100 random realizations of Eve's CSI for each value of $P_A$ are shown. We see that as $N_E$ increases (beyond 8), $\frac{\mathcal{R}_{AE}}{N_E}$ becomes a good approximation of $\frac{R_{AE}}{N_E}$.

\subsection{Power Optimization and Maximum Tolerable Number of Antennas on Eve}
\label{sec:opt_inst_sec}

With $(R_{AB}-\mathcal{R}_{AE})^+$ as the objective function, we can now develop an optimization algorithm to  optimize the power distribution. Note that since the CSI required for $R_{AB}$ is known to Alice and Bob, we do not need to replace $R_{AB}$ by its asymptotic form.

Specifically, we can use this cost function $g(\mathbf{x}_r)\triangleq \mathcal{R}_{AE}-R_{AB}$ where
$\mathbf{x}_r=[\mathbf{q}_r^T, P_n, P_B]^T$. Then, we need to solve the following problem:
\begin{equation}
\label{eq:opt_inst}
\begin{aligned}
& \min_r\min_{\mathbf{x}_r}
& &  g\left (\mathbf{x}_r\right)\\
& \quad \quad \quad \text{s.t.} & &
   \sum_{i = 1}^r \mathbf{q}_r(i)+P_n \leq P_A^{max} \\
& & &  \mathbf{q}_r(i) \geq 0,\ \forall i=1,\dots, r\\
&&& P_n\geq 0\\
& & &  0\leq P_B \leq P_B^{max}.
\end{aligned}
\end{equation}
Here the optimization of $r$ is simple, which can be done via sequential search. For a given $r$, the above problem is not convex.
Although the constraints are convex, the cost $g\left (\mathbf{x}_r\right)$ is not. To see this, let us rewrite this function as follows:
\begin{align}
    &g\left (\mathbf{x}_r\right) =-\log|\mathbf{C}_B+\mathbf{H}\mathbf{V}_1\mathbf{Q}_r\mathbf{V}_1^H\mathbf{H}^H|+\log |\mathbf{C}_B|\notag\\
    & +\sum_{j=1}^{L_{\bar{\mathbf{\Theta}}_3}} \log\left (1+\bar{\eta}_3\left (\bar{\mathbf{\Theta}}_3\right)_{j,j}\right) -\sum_{j=1}^{L_{\bar{\mathbf{\Theta}}_4}} \log\left (1+\bar{\eta}_4\left (\bar{\mathbf{\Theta}}_4\right)_{j,j}\right)\notag\\
    &+ N_E\log \left (\frac{\bar{\eta}_4}{\bar{\eta}_3}\right) +N_E\left (\bar{\eta}_3 -\bar{\eta}_4\right)\log\left (e\right).\label{eq:21}
\end{align}
The non-convex parts of $g\left (\mathbf{x}_r\right)$ are $\log |\mathbf{C}_B|$ and
 $\sum_{j=1}^{L_{\bar{\mathbf{\Theta}}_3}} \log\left (1+\bar{\eta}_3\left (\bar{\mathbf{\Theta}}_3\right)_{j,j}\right)$, which are concave functions of $\mathbf{x}_r$. These two terms can be replaced by their upper bounds based on the first-order Taylor-series expansion around the solution of the previous iteration. Also, the dependence of $g\left (\mathbf{x}_r\right)$ on $\bar{\eta}_3$ and $\bar{\eta}_4$ can be resolved by choosing the values of $\bar{\eta}_3$ and $\bar{\eta}_4$  as follows:
\begin{align}
\label{eq:eq_eta2}
     {1-\bar{\eta}^t_i}=\frac{\bar{\beta}_i\bar{\eta}^t_i}{L_{\bar{\mathbf{\Theta}}_i}}
    \sum_{j=1}^{L_{\bar{\mathbf{\Theta}}_i}} \frac{(\bar{\mathbf{\Theta}}^t_i)_{j,j}}{1+\bar{\eta}^t_i(\bar{\mathbf{\Theta}}^t_i)_{j,j}}, i=3,4.
\end{align}
where  $t$ denotes the $t$th iteration. In other words, at iteration $t$,
the following convex problem is solved:
\begin{equation}
\label{eq:opt_inst2}
\begin{aligned}
& \mathbf{x}_r^{t+1}=\argmin_{\mathbf{x}_r}
& &  h^t\left (\mathbf{x}_r\right)\\
& \quad \quad \quad \quad \quad \text{s.t.} & &
   \sum_{i = 1}^r \mathbf{q}_r(i)+P_n \leq P_A^{max} \\
& & &  \mathbf{q}_r(i) \geq 0,\ \forall i=1,\dots, r\\
&&& P_n\geq 0\\
& & &  0\leq P_B \leq P_B^{max}.
\end{aligned}
\end{equation}
where
\begin{align}
\label{eq:g_bar}
    &h^t\left (\mathbf{x}_r\right) = -\log|\mathbf{C}_B+\mathbf{H}\mathbf{V}_1\mathbf{Q}_r\mathbf{V}_1^H\mathbf{H}^H|\notag\\&-\sum_{j=1}^{L_{\bar{\mathbf{\Theta}}_4}} \log\left (1+\bar{\eta}^t_4\left (\bar{\mathbf{\Theta}}_4\right)_{j,j}\right) +\left (\mathbf{x}_r-\mathbf{x}_r^t\right)^T \nabla_{\mathbf{x}_r}f^t|_{\mathbf{x}_r=\mathbf{x}_r^t},
\end{align}
and $f^t\left (\mathbf{x}_r\right)= \log|\mathbf{C}_B| +\sum_{j=1}^{L_{\bar{\mathbf{\Theta}}_3}} \log\left (1+\bar{\eta}_3^t\left (\bar{\mathbf{\Theta}}_3\right)_{j,j}\right)$.
All constant terms  in  \eqref{eq:21}  are omitted in (\ref{eq:g_bar}) as they do not affect the optimization.

\begin{algorithm}[bt!]
\label{alg:alg1}
\begin{algorithmic}
\STATE Choose proper $\epsilon$, $\bar{\eta}^0_3$, $\bar{\eta}^0_4$, and set $g^{min}=0$.
    \FOR{$r=1:N_A$}
    \STATE Initialize $\mathbf{x}_r$ satisfying the constraints.
    \STATE Set $t=0$.
    \WHILE{$\frac{\|\mathbf{x}_r^t-\mathbf{x}_r^{t-1}\|}{\|\mathbf{x}_r^{t-1}\|}> \epsilon$}\STATE
    Solve  (\ref{eq:opt_inst2}) to get $\mathbf{x}_r^{t+1}$.\\
    Update $\bar{\eta}_3$, $\bar{\eta}_4$ by solving  (\ref{eq:eq_eta2}) using $\mathbf{x}_r^{t+1}$.\\
    $t=t+1.$
    \ENDWHILE
     \IF{$g\left (\mathbf{x}_r^t\right)<g^{min}$}
    \STATE $g^{min}=g\left (\mathbf{x}_r^t\right)$
    \STATE $\mathbf{x}^{min}=\mathbf{x}_r^t$
    \ENDIF
    \ENDFOR
    \STATE Return $\mathbf{x}^{min}$.
    \end{algorithmic}
    \caption{Algorithm for power optimization}
\end{algorithm}

Algorithm 1 details the proposed procedure for the power optimization.
It is worth mentioning that a different optimization approach was explored in our previous work \cite{sohrabi_secrecy}, where a stochastic optimization approach was applied to the objective function $R_{AB}-\mathcal{E}[R_{AE}]$. These two approaches  more or less give the same results, but Algorithm 1 in this paper has a significantly lower complexity.
To illustrate a performance gain of the optimized powers over non-optimal powers, we will not repeat similar figures as available in \cite{sohrabi_secrecy}. But next we consider the maximum tolerable number of antennas on Eve, which can be defined in several ways. One is
\begin{equation}\label{}
  \tilde N_E\triangleq \max N_E, \mbox{ s.t. } R_{AB}-R_{AE}>0.
\end{equation}
which however depends on instantaneous CSI everywhere. Another is
\begin{equation}\label{eq:bar_N_E}
  \bar N_E\triangleq \max N_E, \mbox{ s.t. } \mathcal{R}_{AB}-\mathcal{R}_{AE}>0.
\end{equation}
where $\mathcal{R}_{AB}$ and $\mathcal{R}_{AE}$ are asymptotic forms of $R_{AB}$ and $R_{AE}$ respectively. Obviously, $\bar N_E$ is a function of $\mathbf{x}_r$. The third definition is
\begin{equation}\label{}
  \bar N_E^{opt}\triangleq \max N_E, \mbox{ s.t. } \left (\frac{1}{L}\sum_{l=1}^L (R_{AB,l}^{opt}-\mathcal{R}_{AE,l}^{opt})\right )^+>0.
\end{equation}
where $R_{AB,l}^{opt}-\mathcal{R}_{AE,l}^{opt}$ is a value of $R_{AB}-\mathcal{R}_{AE}$ corresponding to a random realization of $\mathbf{H}$ and $\mathbf{G}$ and the corresponding optimal $\mathbf{x}_{r}$ and $r$, and $L$ is the total number of realizations of $\mathbf{H}$ and $\mathbf{G}$ for each $N_E$.

To obtain $\bar N_E$ in \eqref{eq:bar_N_E}, we will let  $N_A>N_B$, $r=N_B$ and $\mathbf{Q}_r= \frac{P_s}{N_B}\mathbf{I}$. Hence, $range(\mathbf{V}_2)$ is the null-space of $\mathbf{H}$. As a consequence, $\mathbf{H}\mathbf{V}_1\mathbf{Q}_r\mathbf{V}_1^H\mathbf{H}^H=\frac{{P}_S}{N_B}\mathbf{H}\mathbf{H}^H$, and \eqref{eq:R_AB} becomes
\begin{equation}\label{eq:RAB}
    \begin{aligned}
   & R_{AB} = \\&\log |\mathbf{I}+ \frac{\rho {P}_B}{N_B} \mathbf{G}\mathbf{G}^H+\frac{P_s}{N_B}\mathbf{H}\mathbf{H}^H|-\log |\mathbf{I}+ \frac{\rho {P}_B}{N_B} \mathbf{G}\mathbf{G}^H|\\&=
    \log |\mathbf{I}+ \mathbf{J}_1\mathbf{\Theta}_1\mathbf{J}_{1}^H|-\log |\mathbf{I}+ \mathbf{J}_2\mathbf{\Theta}_2\mathbf{J}_2^H|
\end{aligned}
\end{equation}
where $\mathbf{J}_1=\frac{1}{\sqrt{N_B}}\left [\mathbf{H}, \mathbf{G}\right ]$, $\mathbf{J}_2=\frac{1}{\sqrt{N_B}}\mathbf{G}$, $\mathbf{\Theta}_2=\rho P_B \mathbf{I}$, and
$\mathbf{\Theta}_1=\mathrm{diag}\left (\left [{P_s} \mathbf{1}_{N_A},{\rho {P}_B}\mathbf{1}_{N_B}\right ]\right)
$.
Applying the lemma to \eqref{eq:RAB} yields that for $\beta_1=\frac{N_A+N_B}{N_B}$, $\beta_2=1$ and a large $N_B$,
\begin{eqnarray}\label{eq:RAB_a}
  \frac{R_{AB}}{N_B} &\simeq& \Omega\left (\beta_1, \mathbf{\Theta}_1 ,\eta_1\right)-\Omega\left (\beta_2, \mathbf{\Theta}_2 ,\eta_2\right)\nonumber\\
   &=& (\beta_1-1) \log (1+\eta_1 P_s) + \log\frac{1+\eta_1 \rho P_B}{1+\eta_2 \rho P_B} \notag \nonumber\\
   &&  + \log(\frac{\eta_2}{\eta_1})+(\eta_1-\eta_2)\log (e)\nonumber\\
   &\triangleq & \frac{\mathcal{R}_{AB}}{N_B}
\end{eqnarray}
where
$\eta_1$ is the solution of $\eta$ to \eqref{eq:eq_eta} with $\beta=\beta_1$ and $\boldsymbol{\Theta}=\mathbf{\Theta}_1$, which reduces to
\begin{equation}
    1-\eta_1=\frac{(\beta_1-1)\eta_1 P_s}{1+\eta_1 P_s}+\frac{\eta_1 \rho P_B}{1+\eta_1 \rho P_B}.
\end{equation}
and $\eta_2$ is the solution to
\begin{equation}
    1-\eta_2=\frac{\eta_2 \rho P_B}{1+\eta_2 \rho P_B},
\end{equation}
or equivalently $\eta_2=\frac{\sqrt{1+4\rho P_B}-1}{2\rho P_B}$.

Also with $r=N_B$ and $\mathbf{Q}_r= \frac{P_s}{N_B}\mathbf{I}$,
 $\mathcal{R}_{AE}$ in
 \eqref{eq:RAE_b}  reduces to
\begin{eqnarray}\label{eq:RAE_a}
  \frac{R_{AE}}{N_E} &\simeq& \Omega\left (\beta_3, \mathbf{\Theta}_3 ,\eta_3\right)-\Omega\left (\beta_4, \mathbf{\Theta}_4 ,\eta_4\right) \nonumber \\
   &=& (\beta_3-\beta_4)\log \left(1+\frac{a P_s \eta_3}{\beta_3-\beta_4}\right)\nonumber\\
   &&+(\beta_3-\beta_4)\log \frac{\beta_3-\beta_4+ b P_B \eta_3}{\beta_3-\beta_4+ b P_B \eta_4}\nonumber\\
   &&+(2\beta_4-\beta_3)\log \frac{2\beta_4-\beta_3+ a P_n \eta_3}{2\beta_4-\beta_3+ a P_n \eta_4}\nonumber\\
    &&+\log \left(\frac{\eta_4}{\eta_3}\right) + \left(\eta_3-\eta_4\right)\log \left(e\right)\nonumber\\
    &\triangleq & \frac{\mathcal{R}_{AE}}{N_E}
\end{eqnarray}
where

 \begin{equation}
     \mathbf{\Theta}_3=N_E\mathrm{diag}\left (\left [\frac{aP_s}{N_B} \mathbf{1}_{N_B},\frac{aP_n}{N_A-N_B}\mathbf{1}_{N_A-N_B},\frac{b P_B}{N_B} \mathbf{1}_{N_B}\right ]\right),
\end{equation}
and
\begin{equation}
    \mathbf{\Theta}_4=N_E\mathrm{diag}\left (\left [\frac{aP_n}{N_A-N_B}\mathbf{1}_{N_A-N_B},\frac{b P_B}{N_B} \mathbf{1}_{N_B}\right ]\right),
\end{equation}
also
$\beta_3=\frac{N_A+N_B}{N_E}$, $\beta_4=\frac{N_A}{N_E}$, $\eta_3$  is the solution to
\begin{eqnarray}\label{eq:delta_3}
    1-\eta_3&=& \frac{aP_s\eta_3}{1+aP_s\eta_3\frac{1}{\beta_3-\beta_4}}+
   \frac{aP_n\eta_3}{1+aP_n\eta_3\frac{1}{2\beta_4-\beta_3}}\nonumber\\
   && +\frac{bP_B\eta_3}{1+bP_B\eta_3\frac{1}{\beta_3-\beta_4}},
\end{eqnarray}
and $\eta_4$ is the solution to
\begin{eqnarray}\label{eq:delta_4}
     1-\eta_4 &=& \frac{aP_n\eta_4}{1+aP_n\eta_4\frac{1}{2\beta_4-\beta_3}}\nonumber\\
   &&+\frac{bP_B\eta_4}{1+bP_B\eta_4\frac{1}{\beta_3-\beta_4}} .
\end{eqnarray}

Fig. \ref{fig:zero_sec_comp_pa} shows $\bar N_E$ versus $N_A$ and $\bar N_E^{opt}$ versus $N_A$ where $P_A^{max}=P_B^{max}=30dB$ and $N_A=2N_B$ ($\beta_1=3$). For  $\bar N_E$, we also chose $P_s=P_n=\frac{P_A}{2}$ and $P_A=P_B$.
We see that $\bar N_E^{opt}$ is consistently larger than $\bar N_E$. And the gap between the two is due to the power optimization.

\begin{figure}[bt!]
    \centering
    \includegraphics[width=1\linewidth]{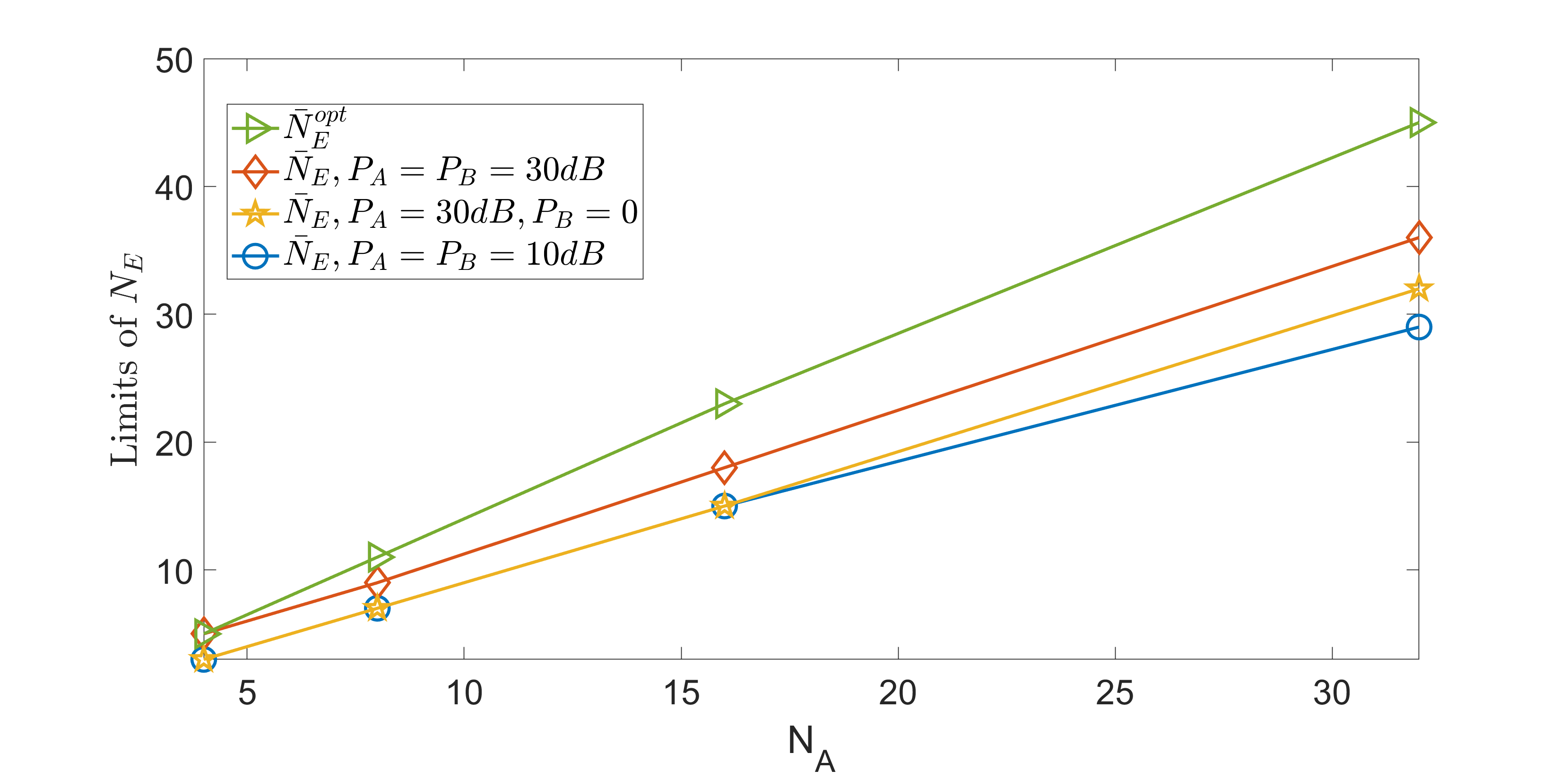}
    \caption{Comparison of  $\bar{N}_E$ and $\bar N_E^{opt}$  vs $N_A$.}
    \label{fig:zero_sec_comp_pa}
\end{figure}

During simulation, we also observed that the optimal $P_s$ is often distributed approximately equally between different streams, and that if $N_E$ gets larger, the optimization favors  smaller $P_B$ and smaller $r$ (the latter of which is consistent with a result in \cite{7740063} which does not use full-duplex jamming at Bob).

With the same parameters as in Fig. \ref{fig:zero_sec_comp_pa}, Fig. \ref{fig:conv_ci} illustrates a convergence property of Algorithm 1 where the mean number of iterations needed for convergence versus $N_A$ is shown. Also shown in Fig. \ref{fig:conv_ci} is the 95\% confidence interval of the number of iterations needed for convergence versus $N_A$. We used 100 random realizations of the channels for each value of $N_A$. The threshold $\epsilon$ used for convergence was chosen to be 0.01.

\begin{figure}[bt!]
    \centering
    \includegraphics[width=1\linewidth]{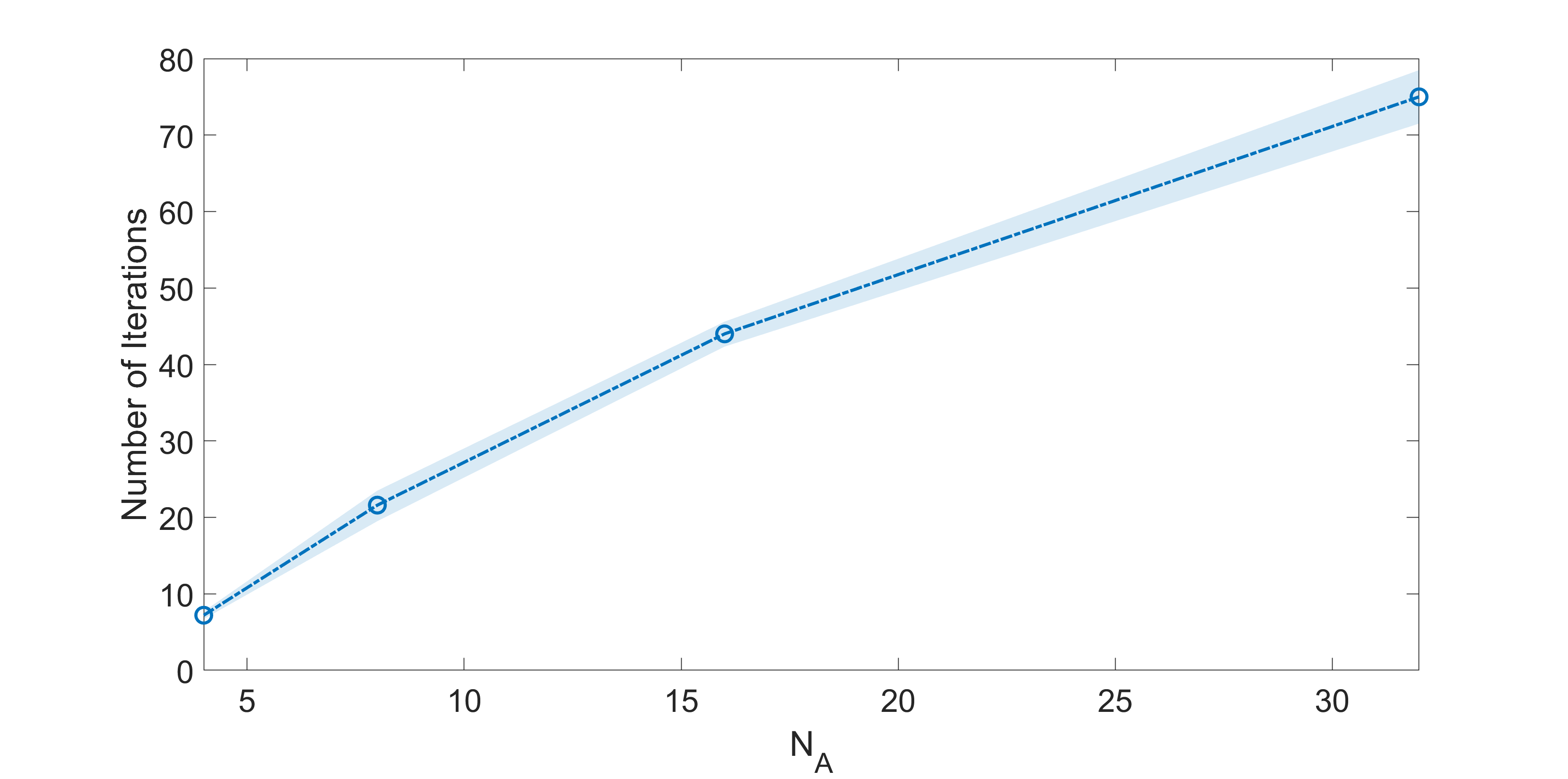}
    \caption{95\% confidence interval of the number of iterations needed for  convergence of Algorithm 1 vs. $N_A$. The dark line is the mean of the number of iterations.}
    \label{fig:conv_ci}
\end{figure}

\subsection{When the Number of Antennas on Eve is Very Large}

We now consider the case where $N_E\gg N_A> N_B$, $r=N_B$ and $\mathbf{Q}_r= \frac{P_s}{N_B}\mathbf{I}$.
It follows that $\beta_3\ll 1$ and $\beta_4\ll 1$. Hence, \eqref{eq:delta_3} implies $1-\eta_3\approx \beta_3$ and \eqref{eq:delta_4} implies $1-\eta_4\approx\beta_4$. Furthermore, referring to the terms in \eqref{eq:RAE_a}, we have
\begin{eqnarray}
  &&\lim_{N_E\rightarrow \infty} N_E(\beta_3-\beta_4)\log \left(1+\frac{a P_s \eta_3}{\beta_3-\beta_4}\right)\nonumber \\
   &=&N_B \log \left (1+\frac{N_EaP_s}{N_B} \right ),
\end{eqnarray}
\begin{eqnarray}\label{}
 && \lim_{N_E\rightarrow \infty} N_E(\beta_3-\beta_4)\log \frac{\beta_3-\beta_4+ b P_B \eta_3}{\beta_3-\beta_4+ b P_B \eta_4}\nonumber\\
  &=&N_B\log 1=0,
\end{eqnarray}
\begin{eqnarray}\label{}
 && \lim_{N_E\rightarrow \infty} N_E(2\beta_4-\beta_3)\log \frac{2\beta_4-\beta_3+ a P_n \eta_3}{2\beta_4-\beta_3+ a P_n \eta_4}\nonumber\\
  &=&(N_A-N_B)\log 1 =0,
\end{eqnarray}
\begin{eqnarray}
  && \lim_{N_E\rightarrow \infty} N_E \left (\log \left(\frac{\eta_4}{\eta_3}\right) + \left(\eta_3-\eta_4\right)\log \left(e\right)\right ) \nonumber\\
   &=& \lim_{N_E\rightarrow \infty} N_E\left (\log \frac{1-\beta_4}{1-\beta_3}+(\beta_4-\beta_3)\log e\right )\nonumber\\
   &=& \lim_{N_E\rightarrow \infty} N_E \log\left (1+\frac{N_B}{N_E-(N_A+N_B)} \right )-N_B\log e \nonumber\\
   &=& \lim_{N_E\rightarrow \infty} N_E  \frac{ N_B}{N_E-(N_A+N_B)}\log e-N_B\log e\nonumber\\
   &=&0.
\end{eqnarray}
The above equations imply that all terms, except the first, in $\mathcal{R}_{AE}$ from \eqref{eq:RAE_a} converge to zero. Therefore,
\begin{equation}\label{eq:app_Rae}
  \lim_{N_E\rightarrow \infty}R_{AE} =\lim_{N_E\rightarrow \infty}\mathcal{R}_{AE}= N_B\log\left (1+\frac{N_EaP_s}{N_B} \right )\triangleq \mathcal{R}^*_{AE}
\end{equation}
which is independent of $P_n$ and $P_B$ (and hence the optimal $P_n$ and $P_B$ are now zero).
This result implies that if Eve has an unlimited number of antennas then the jamming noise from either Alice or Bob has virtually no impact on Eve's capacity to receive the information from Alice. Furthermore, we see that $R_{AE}$ increases without upper bound as $N_E$ increases while $R_{AB}$ stays independent of $N_E$ (for large $N_E$).

Fig. \ref{fig:eve_approx_asympt} compares $\mathcal{R}^*_{AE}$ from \eqref{eq:app_Rae} with $\mathcal{R}_{AE}$ from \eqref{eq:RAE_a} where $N_A=2N_B=8$ and $P_A = P_B=2P_s=2P_n$. We see that the two results are very close when $N_E>40$.

\begin{figure}
    \centering
    \includegraphics[width = \linewidth]{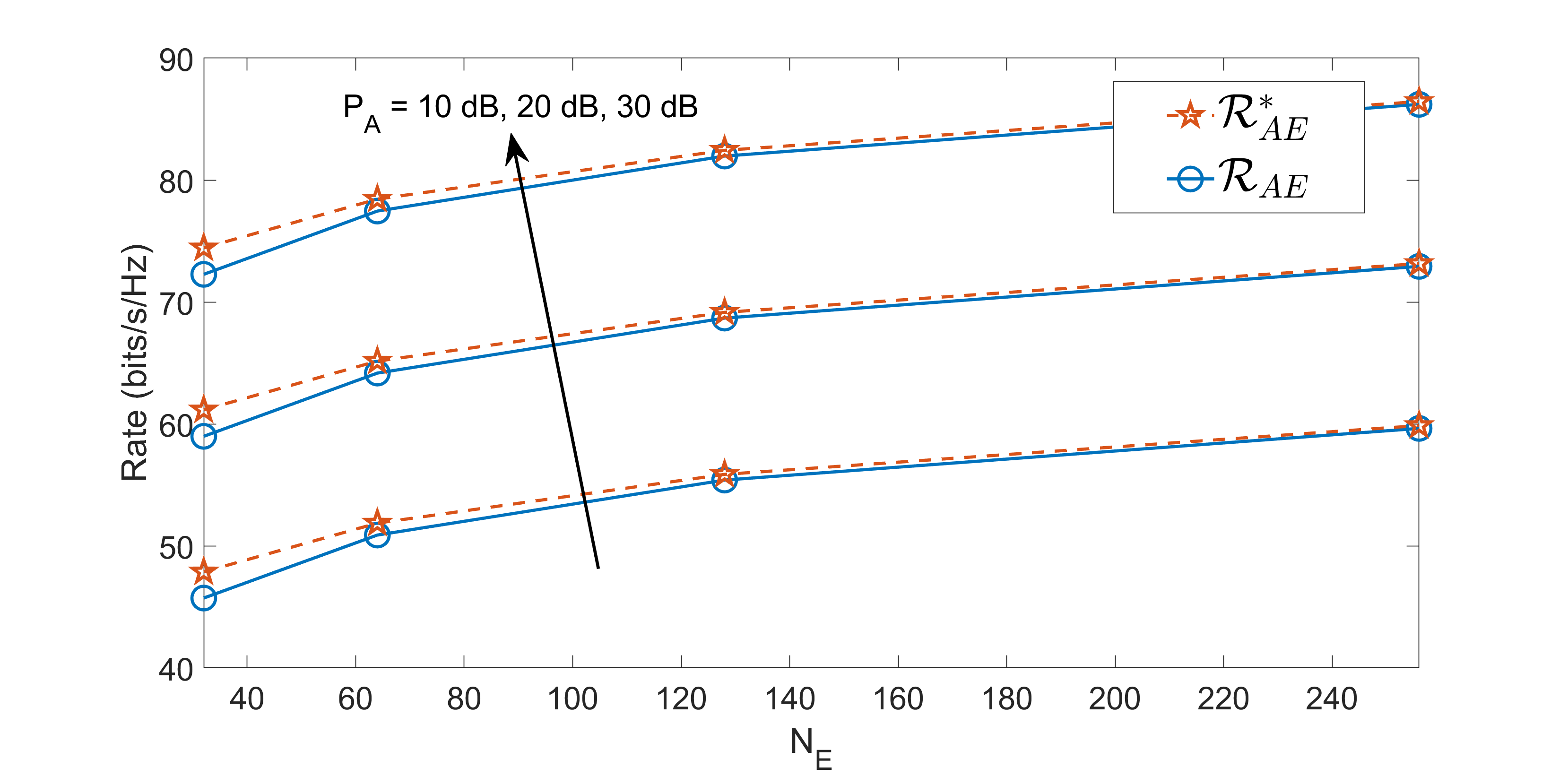}
    \caption{The convergence of $\mathcal{R}_{AE}$ to $\mathcal{R}^*_{AE}$}
    \label{fig:eve_approx_asympt}
\end{figure}

\section{Analysis of ANECE}\label{sec:anece}
A key observation from the previous section is that if Eve knows its CSI and the number of antennas on Eve is large, then neither the artificial noise from multi-antenna Alice nor the full-duplex jamming from multi-antenna Bob can rescue Alice and Bob from being totally exposed to Eve. (This observation is an extension of a previous observation for single-antenna users shown in \cite{hua2018advanced}.)
To handle Eve with large number of antennas, there is a two-phase method involving anti-eavesdropping channel estimation (ANECE) proposed in \cite{hua2018advanced}:
in phase 1 the users conduct ANECE which allows users to obtain their CSI but denies Eve the same ability; and in phase 2 the users transmit information to each other with Eve not knowing its CSI. Both phases are within a common coherence period. While the earlier work has shown promising properties of ANECE, the understanding of ANECE is still incomplete. In this section, we show two new analyses of the secrecy rate of a two-user MIMOME network assisted by ANECE.

To simplify the problem, we do not consider the artificial noise from either Alice or Bob. The first analysis assumes a (globally known) statistical model for all CSI in the network. And the analysis is based on ideal full-duplex devices where there is no self-interference. When a result of this analysis is applied to practice, one must restrict the application to situations where the residual self-interference is negligible. Typically, the residual self-interference is proportional to the transmitted power which increases with the distance between devices. So, a situation where the residual self-interference is negligible corresponds generally to a short-range communication. The second analysis assumes that Eve does not know the statistical distribution of its CSI but rather assumes that Eve is able to perform blind detection of the information from Alice. These two analyses constitute an important new understanding of ANECE, which is not available elsewhere.

A theory where Eve knows the statistical distribution of its CSI can be applicable to situations where Eve's CSI is statistically stationary and experiences many cycles of coherence periods in a time window of interest. A theory where Eve does not know its CSI distribution can be applicable to situations where Eve's CSI is statistically un-stationary  in a time window of interest. Both assumptions have their own merits.

\subsection{Eve uses a statistical model of its CSI}
 Consider a block Rayleigh fading channel for which Alice and Bob first conduct ANECE by transmitting their pilot signals $\mathbf{p}_A(k)$ and $\mathbf{p}_B(k)$ concurrently (in full-duplex mode) where $k=1,\cdots,K_1$ ($K_1$ is the length of the pilot), and then transmit information to each other (over $K_2$ samples). For information transmission, we will consider a one-way transmission and a two-way transmission separately.

  \subsubsection{Channel estimation}
  Define $\mathbf{P}_i=[\mathbf{p}_i(1),\cdots,\mathbf{p}_i(K_1)]$ where $i=A,B$. then the corresponding signals received by Alice, Bob and Eve can be expressed as
 	\begin{subequations}\label{ph1}
		\begin{align}
		&\mathbf{Y}_{A} = \mathbf{H}^{T}\mathbf{P}_{B} + \mathbf{N}_{A}\\
		&\mathbf{Y}_{B} = \mathbf{H}\mathbf{P}_{A} + \mathbf{N}_{B}\label{ph1.b}\\
		&\mathbf{Y}_{E} = \sqrt{a}\mathbf{A}\mathbf{P}_{A} + \sqrt{b}\mathbf{B}\mathbf{P}_{B} + \mathbf{N}_{E}\label{ph1.e}
		\end{align}
	\end{subequations}
where $\mathbf{H}$ is the reciprocal channel matrix between Alice and Bob, and all the noise matrices consist of i.i.d. $\mathcal{CN}(0,1)$.   Here, the self-interferences at Alice and Bob are assumed to be negligible.

It is known and easy to show that for the best performance of the maximum likelihood (ML) estimation (or the MMSE estimation as shown later)  of $\mathbf{H}$ by Bob, $\mathbf{P}_A$ should  be such that $\mathbf{P}_A\mathbf{P}_A^H=\frac{K_1P_A}{N_A}\mathbf{I}_{N_A}$. Similarly, $\mathbf{P}_B$ should be such that $\mathbf{P}_B\mathbf{P}_B^H=\frac{K_1P_B}{N_B}\mathbf{I}_{N_B}$.

In the following analysis, we assume that $\mathbf{H}$, $\mathbf{A}$ and $\mathbf{B}$ all consist of i.i.d. $\mathcal{CN}(0,1)$ elements (from one coherence block to another). This statistical model along with the large-scale fading factors $a$ and $b$ is assumed to be known to everyone.

  Without loss of generality, let
  $N_{A}\geq N_{B}$. Without affecting the channel estimation performance at Alice and Bob, but maximizing the difficulty of channel estimation for Eve, we let the row span of $\mathbf{P}_{B}$ be part of the row span of $\mathbf{P}_{A}$. More specifically, we can write $\mathbf{P}_{A} = \sqrt{\frac{K_{1}P_{A}}{N_{A}}}[\mathbf{I}_{N_{A}},\mathbf{0}_{N_{A}\times(K_{1}-N_{A})}]
   \boldsymbol{\Gamma}$ and $\mathbf{P}_{B} = \sqrt{\frac{K_{1}P_{B}}{N_{B}}}[\mathbf{I}_{N_{B}},\mathbf{0}_{N_{B}\times(K_{1}-N_{B})}]
   \boldsymbol{\Gamma} $ where $\boldsymbol{\Gamma}$ can be any $K_1\times K_1$ unitary matrix. In this way, any estimates of $\mathbf{A}$ and $\mathbf{B}$ by Eve, denoted by $\mathbf{\hat A}$ and $\mathbf{\hat B}$, are ambiguous in that $[\sqrt{a}\mathbf{\hat A},\sqrt{b}\mathbf{\hat B}]$ can be added to $\boldsymbol{\Theta}[\mathbf{C}_A,\mathbf{C}_B]$ without affecting Eve's observation $\mathbf{Y}_E$ where $\boldsymbol{\Theta}\in \mathbb{C}^{N_E\times N_B}$ is arbitrary and $[\mathbf{C}_A,\mathbf{C}_B][\mathbf{P}_{A}^T,\mathbf{P}_{B}^T]^T=0$.

	Let $\mathbf{h} = vec(\mathbf{H})$,  $\mathbf{a} = vec(\mathbf{A})$, $\mathbf{b} = vec(\mathbf{B})$, $\mathbf{y}_{A} = vec(\mathbf{Y}_{A}^{T})$, $\mathbf{y}_{B} = vec(\mathbf{Y}_{B})$, $\mathbf{n}_{A} = vec(\mathbf{N}_{A}^{T})$ and $\mathbf{n}_B = vec(\mathbf{N}_B)$. Note $vec(\mathbf{X}\mathbf{Y}\mathbf{Z})=(\mathbf{Z}^T\otimes\mathbf{X})vec(\mathbf{Y})$. Then \eqref{ph1} becomes
 	\begin{subequations}\label{ph1.1}
		\begin{align}
		&\mathbf{y}_{A} = (\mathbf{I}_{N_A} \otimes \mathbf{P}_{B}^{T}) \mathbf{h} + \mathbf{n}_{A}\\
		&\mathbf{y}_{B} = (\mathbf{P}_{A}^{T} \otimes \mathbf{I}_{N_B}) \mathbf{h} + \mathbf{n}_{B}\label{ph1.1.b}\\
		&\mathbf{y}_{E} = \sqrt{a}(\mathbf{P}_{A}^{T}\otimes \mathbf{I}_{N_E}) \mathbf{a} + \sqrt{b}(\mathbf{P}_{B}^{T}\otimes \mathbf{I}_{N_E}) \mathbf{b} + \mathbf{n}_{E}.\label{ph1.1.e}
		\end{align}
	\end{subequations}

It is known that the minimum-mean-squared-error (MMSE) estimate of a vector $\mathbf{x}$ from another vector $\mathbf{y}$ is $\mathbf{\hat x}=\mathbf{K}_{\mathbf{x},\mathbf{y}}\mathbf{K}_{\mathbf{y}}^{-1}\mathbf{y}$ with $\mathbf{K}_{\mathbf{x},\mathbf{y}}=\mathcal{E}\{\mathbf{x}\mathbf{y}^H\}$ and $\mathbf{K}_{\mathbf{y}}=\mathcal{E}\{\mathbf{y}\mathbf{y}^H\}$. And the error $\Delta\mathbf{x}=\mathbf{x}-\mathbf{\hat x}$ has the covariance matrix $\mathbf{K}_{\Delta\mathbf{x}}=\mathbf{K}_{\mathbf{x}}-\mathbf{K}_{\mathbf{x},\mathbf{y}}
\mathbf{K}_{\mathbf{y}}^{-1}\mathbf{K}_{\mathbf{x},\mathbf{y}}^H$.

Let $\mathbf{\hat h}_A$ be the MMSE estimate of $\mathbf{h}$ by Alice, and $\Delta\mathbf{h}_A = \mathbf{h}-\mathbf{\hat h}_A$ be its error. Similar notations are defined for Bob and Eve.
It is easy to show that the covariance matrices of the errors of these estimates are, respectively, $\mathbf{K}_{\Delta{\mathbf{h}}_A}  = \sigma^{2}_A\mathbf{I}_{N_AN_B}$, $\mathbf{K}_{\Delta{\mathbf{h}}_B}  = \sigma^{2}_B\mathbf{I}_{N_AN_B}$, $\mathbf{K}_{\Delta{\mathbf{a}}}
	= \sigma^{2}_{EA}\mathbf{I}_{N_AN_E}$ and $\mathbf{K}_{\Delta{\mathbf{b}}}
	= \sigma^{2}_{EB}\mathbf{I}_{N_BN_E}$ where $\sigma^{2}_A = \frac{1}{1 +  K_{1}P_B/N_B}$,   $\sigma^{2}_B = \frac{1}{1 +  K_{1}P_A/N_A}$, $\sigma^{2}_{EA} = \frac{bK_{1}P_{B}/N_{B} + 1}{(aK_{1}P_{A}/N_{A} + bK_{1}P_{B}/N_{B}) +1}$ and $\sigma^{2}_{EB} = \frac{aK_{1}P_{A}/N_{A} + 1}{(aK_{1}P_{A}/N_{A} + bK_{1}P_{B}/N_{B}) +1}$.
	
	\subsubsection{One-way information transmission}
	Now assume that following the pilots (over $K_1$ samples) transmitted by Alice and Bob in full-duplex mode, Alice transmits information (over $K_2$ samples) to Bob in half-duplex mode. Namely, while the first phase is in full-duplex, the second phase is in half-duplex. In the second phase, Bob and Eve receive
	\begin{equation}\label{ph2}
	\begin{aligned}
	&\mathbf{Y}_{B} = \mathbf{H}\mathbf{S}_{A} + \mathbf{N}_{B}\\
	&\mathbf{Y}_{E} = \sqrt{a}\mathbf{A}\mathbf{S}_{A}  + \mathbf{N}_{E}
	\end{aligned}
	\end{equation}
	where $\mathbf{S}_{A} = [\mathbf{s}_{A}(1),\dots, \mathbf{s}_{A}(K_{2})]$. The corresponding vector forms of the above are
	\begin{subequations}\label{ph2.1}
		\begin{align}
		&\mathbf{y}_{B} = (\mathbf{I}_{K_2}\otimes\mathbf{H})\bar{\mathbf{s}}_{A} + \mathbf{n}_{B}\\
		&\mathbf{y}_{E} = \sqrt{a}(\mathbf{I}_{K_2}\otimes\mathbf{A})\bar{\mathbf{s}}_{A} + \mathbf{n}_{E}\label{ph2.1.e}
		\end{align}
	\end{subequations}
	where $\bar{\mathbf{s}}_{A} = vec(\mathbf{S}_{A})$ (which is assumed to be independent of all channel parameters). Then an achievable secrecy rate in bits/s/Hz in phase 2 from Alice to Bob (conditional on the MMSE channel estimation in phase 1) is
	\begin{equation}\label{eq:R_one}
	\mathcal{R}_{one} = \frac{1}{K_2} \big(I(\bar{\mathbf{s}}_{A};\mathbf{y}_{B}|\hat{\mathbf{h}}_{B}) - I(\bar{\mathbf{s}}_{A};\mathbf{y}_{E}|\hat{\mathbf{a}})\big)^{+}
	\end{equation}
	
To analyze $\mathcal{R}_{one}$, we now assume
	$P_{A} = P_{B} = P$ (which holds for both phases 1 and 2) and that $\mathbf{s}_{A}(k)$ are i.i.d. with $\mathcal{CN}(0, \frac{P_{A}}{N_{A}}\mathbf{I}_{N_A})$. We also use $\hat{\mathbf{H}}_{B} = ivec(\hat{\mathbf{h}}_{B})\in\mathbb{C}^{N_B\times N_A}$ (i.e., $\hat{\mathbf{h}}_{B}=vec(\hat{\mathbf{H}}_{B})$).

We will next derive lower and upper bounds on $\mathcal{R}_{one}$. To do that, we need to obtain lower and upper bounds on $I(\bar{\mathbf{s}}_{A};\mathbf{y}_{B}|\hat{\mathbf{h}}_{B})$ and those on $I(\bar{\mathbf{s}}_{A};\mathbf{y}_{E}|\hat{\mathbf{a}})$.

First, we have
\begin{eqnarray}\label{eq:111}
  I(\bar{\mathbf{s}}_{A};\mathbf{y}_{B}|\hat{\mathbf{h}}_{B})&=&h(\bar{\mathbf{s}}_{A}|\hat{\mathbf{h}}_{B}) - h(\bar{\mathbf{s}}_{A}|\mathbf{y}_{B},\hat{\mathbf{h}}_{B})\nonumber \\
  &=&
h(\bar{\mathbf{s}}_{A})-h(\bar{\mathbf{s}}_{A}|\mathbf{y}_{B},\hat{\mathbf{h}}_{B}).
\end{eqnarray}
It is known that $h(\bar{\mathbf{s}}_{A})=\log\left [(\pi e)^{N_AK_2}\left |\frac{P_A}{N_A}\mathbf{I}_{N_AK_2}\right |\right ]$. It is also known \cite{ElGamal2011} that for a random vector $\mathbf{s}\in \mathbb{C}^{n\times 1}$ and another random vector $\mathbf{w}$, $h(\mathbf{s}|\mathbf{w})\leq \log \left [(\pi e)^n |\mathbf{K}_{\mathbf{s}|\mathbf{w}}| \right ]$ where $\mathbf{K}_{\mathbf{s}|\mathbf{w}}=\mathbf{K}_{\mathbf{s}}-
\mathbf{K}_{\mathbf{s},\mathbf{w}}(\mathbf{K}_{\mathbf{w}})^{-1}\mathbf{K}_{\mathbf{s},\mathbf{w}}$ which is the covariance matrix of the MMSE estimation of $\mathbf{s}$ from $\mathbf{w}$. Note that $\mathbf{y}_B = (\mathbf{I}_{K_2}\otimes \mathbf{\hat H}_B)\mathbf{\bar s}_A+(\mathbf{I}_{K_2}\otimes \Delta\mathbf{H}_B)\mathbf{\bar s}_A+\mathbf{n}_B$. Then conditional on $\mathbf{\hat H}_B$ (which is independent of $\mathbf{\bar s}_A$), the covariance matrix of the MMSE estimate of $\mathbf{\bar s}_A$ from $\mathbf{y}_B$ is
$\mathbf{K}_{\bar{\mathbf{s}}_{A}|\mathbf{y}_{B},\hat{\mathbf{h}}_{B}} = \frac{P_{A}}{N_{A}}\mathbf{I}_{N_AK_2} - \frac{P_{A}^{2}}{N_{A}^{2}}(\mathbf{I}_{K_2}\otimes\hat{\mathbf{H}}_{B}^{H})
(\frac{P_{A}}{N_{A}}(\mathbf{I}_{K_2}\otimes\hat{\mathbf{H}}_{B}\hat{\mathbf{H}}_{B}^{H}) + \mathbf{K}_{B} + \mathbf{I}_{N_BK_2})^{-1}(\mathbf{I}_{K_2}\otimes\hat{\mathbf{H}}_{B})$ where $\mathbf{K}_{B} = \mathcal{E}\{(\mathbf{I}_{K_2}\otimes\Delta \mathbf{H}_{B})\bar{\mathbf{s}}_{A}\bar{\mathbf{s}}_{A}^{H}(\mathbf{I}\otimes\Delta\mathbf{H}^{H}_{B})\} = \frac{ P_{A}}{1 +  K_{1}P_{A}/N_{A}}\mathbf{I}_{N_BK_2}$. Using $|\mathbf{I}_{r_A}+\mathbf{A}\mathbf{B}|=|\mathbf{I}_{r_B}+\mathbf{B}\mathbf{A}|$ where $r_A$ and $r_B$ are the numbers of rows of $\mathbf{A}$ and $\mathbf{B}$ respectively, one can verify that $\log|\mathbf{K}_{\bar{\mathbf{s}}_{A}|\mathbf{y}_{B},\hat{\mathbf{h}}_{B}}|=
N_AK_2\log\frac{P_A}{N_A}+\log|\mathbf{K}_B+\mathbf{I}_{N_BK_2}|-
\log|\frac{P_A}{N_A}(\mathbf{I}_{K_2}\otimes
\mathbf{\hat H}_B\mathbf{\hat H}_B^H)+\mathbf{K}_B+\mathbf{I}_{N_BK_2}|=N_AK_2\log\frac{P_A}{N_A}-K_2\log|\mathbf{I}_{N_B}+
\frac{P_A/N_A}{1+\frac{P_A}{1+K_1P_A/N_A}}\mathbf{\hat H}_B\mathbf{\hat H}_B^H|$.
Applying the above results to \eqref{eq:111} yields
	\begin{equation}\label{legiR1}
		\begin{aligned}
		&I(\bar{\mathbf{s}}_{A};\mathbf{y}_{B}|\hat{\mathbf{h}}_{B})\\
		&\geq \log|\frac{P_{A}}{N_{A}}\mathbf{I}_{N_AK_2}|  - \mathcal{E}\{\log|\mathbf{K}_{\bar{\mathbf{s}}_{A}|\mathbf{y}_{B},
\hat{\mathbf{h}}_{B}}|\}\\
		& = K_{2}\mathcal{E}\{\log|\mathbf{I}_{N_B} + \frac{P_{A}/N_{A}}{1 + \frac{ P_{A}}{1 + K_{1}P_{A}/N_{A} }}\hat{\mathbf{H}}_{B}\hat{\mathbf{H}}_{B}^{H}|\}\\
		&\triangleq\mathcal{R}_{B}^{-}.
		\end{aligned}
	\end{equation}

To derive an upper bound on $I(\bar{\mathbf{s}}_{A};\mathbf{y}_{B}|\hat{\mathbf{h}}_B)$, we now write
\begin{equation}\label{eq:112}
  I(\bar{\mathbf{s}}_{A};\mathbf{y}_{B}|\hat{\mathbf{h}}_B)= h(\mathbf{y}_{B}|\hat{\mathbf{h}}_{B}) - h(\mathbf{y}_{B}|\hat{\mathbf{h}}_{B},\bar{\mathbf{s}}_{A}).
\end{equation}
Here, $h(\mathbf{y}_{B}|\hat{\mathbf{h}}_{B})\leq \mathcal{E} \{ \log[ (\pi e)^{N_BK_2} | \frac{P_A}{N_A}(\mathbf{I}_{K_2}\otimes \mathbf{\hat H}_B\mathbf{\hat H}_B^H)+\mathbf{K}_B+\mathbf{I}_{N_BK_2} |] \}=K_2\mathcal{E}\{ \log[(\pi e)^{N_B} | \frac{P_A}{N_A}(\mathbf{\hat H}_B\mathbf{\hat H}_B^H)+(1+\frac{P_A}{1+K_1P_A/N_A})\mathbf{I}_{N_B} |]\}$, and $h(\mathbf{y}_{B}|\hat{\mathbf{h}}_{B},\bar{\mathbf{s}}_{A})=\mathcal{E} \{ \log [(\pi e)^{N_BK_2} | \frac{1}{1+K_1P_A/N_A}(\mathbf{S}_A^T\mathbf{S}_A^*\otimes \mathbf{I}_{N_B})+\mathbf{I}_{N_BK_2} |] \}=N_B\mathcal{E} \{ \log [(\pi e)^{K_2} | \frac{1}{1+K_1P_A/N_A}(\mathbf{S}_A^T\mathbf{S}_A^*)+\mathbf{I}_{K_2} | ]  \}$. Note that conditional on $\hat{\mathbf{h}}_{B}$ and $\bar{\mathbf{s}}_{A}$ the covariance matrix of $\mathbf{y}_{B}$ is invariant to $\hat{\mathbf{h}}_{B}$.  Now
 define
	\begin{equation}
	\mathbf{M}_{A} = \left\{
	\begin{aligned}
	&\frac{N_{A}}{P_{A}}\mathbf{S}_{A}^{T}\mathbf{S}_{A}^{*},~K_2<N_{A} \\
	&\frac{N_{A}}{P_{A}}\mathbf{S}_{A}^{*}\mathbf{S}_{A}^{T},~K_2\geq N_{A}
	\end{aligned}\right.
	\end{equation}
which is a full rank matrix for any $N_A$ and $K_2$ and a self-product of $\sqrt{\frac{N_A}{P_A}}\mathbf{S}_A$ with i.i.d. $\mathcal{CN}(0,1)$ entries. Also define $t_A=\min\{N_A,K_2\}$ and $r_A=\max\{N_A,K_2\}$. It follows that (as part of $h(\mathbf{y}_{B}|\hat{\mathbf{h}}_{B},\bar{\mathbf{s}}_{A})$)
	\begin{subequations}\label{ra++}
		\begin{align}
&\mathcal{E} \{ \log  | \frac{1}{1+K_1P_A/N_A}(\mathbf{S}_A^T\mathbf{S}_A^*)+\mathbf{I}_{K_2} |   \}\notag\\
&=
		\mathcal{E}\{\log| \frac{ P_{A}/N_{A}}{1 +  K_{1}P_{A}/N_{A}}\mathbf{M}_{A} +\mathbf{I}_{t_A}|\}\notag\\
		& \geq t_{A}\mathcal{E}\{\log(1 + | \frac{ P_{A}/N_{A}}{1 +  K_{1}P_{A}/N_{A}}\mathbf{M}_{A}|^{\frac{1}{t_{A}}})\}\label{ra++.1}\\
		& = t_{A}\mathcal{E}\big\{\log(1 + \frac{ P_{A}/N_{A}}{1 +  K_{1}P_{A}/N_{A}}\exp(\frac{1}{t_{A}}\ln| \mathbf{M}_{A}|)\big)\big\}\notag\\
		& \geq t_{A}\log\big(1 + \frac{ P_{A}/N_{A}}{1 +  K_{1}P_{A}/N_{A}}\exp(\frac{1}{t_{A}}\mathcal{E}\{\ln| \mathbf{M}_{A}|\})\big)\label{ra++.2}\\
		& = t_{A}\log\big(1 + \frac{ P_{A}/N_{A}}{1 +  K_{1}P_{A}/N_{A}}\exp(\frac{1}{t_{A}}\sum_{j=1}^{t_{A}}\sum_{k=1}^{r_{A}-j}\frac{1}{k} - \gamma)\big)\label{ra++.3}
		\end{align}
	\end{subequations}
	where \eqref{ra++.1} is due to the matrix Minkowski’s inequality  $|\mathbf{X}+\mathbf{Y}|^{1/n}\geq |\mathbf{X}|^{1/n}+|\mathbf{Y}|^{1/n}$ where $\mathbf{X}$ and $\mathbf{Y}$ are $n\times n$ positive definite matrices \cite{Horn2012a}, \eqref{ra++.2} is due to the Jensen's inequality and that $\log(1 + ae^{x})$ is a convex function of $x$ when $a >0$,  and \eqref{ra++.3} is based on \cite[Th.1]{Oyman2003} where $\gamma \approxeq 0.57721566$ is Euler's constant.  Defining $e_{A} = \exp(\frac{1}{t_{A}}\sum_{j=1}^{t_{A}}\sum_{k=1}^{r_{A}-j}\frac{1}{k} - \gamma)$ and applying the above results since \eqref{eq:112}, we have from \eqref{eq:112} that
	\begin{equation}\label{bobR}
	\begin{aligned}
	&I(\bar{\mathbf{s}}_{A};\mathbf{y}_{B}|\hat{\mathbf{h}}_B) \\
	&\leq K_{2}\mathcal{E}\{\log| \mathbf{I}_{N_B} +  \frac{P_{A}/N_{A}\hat{\mathbf{H}}_{B}\hat{\mathbf{H}}_{B}^{H}}{1 + \frac{P_{A} }{1 +  K_{1}P_{A}/N_{A}}} |\}\\
	&\quad  + N_{B}\log\bigg(\frac{(1 + \frac{ P_{A}}{1 +  K_{1}P_{A}/N_{A}})^{K_{2}}}{\big(1 + \frac{ P_{A}/N_{A}}{1 +  K_{1}P_{A}/N_{A}}e_{A}\big)^{t_{A}}}\bigg)\\
	&\triangleq\mathcal{R}_{B}^{+}
	\end{aligned}
	\end{equation}
	From \eqref{legiR1} and \eqref{bobR} we see that the difference between the upper and lower bounds on $I(\bar{\mathbf{s}}_{A};\mathbf{y}_{B}|\hat{\mathbf{h}}_B)$ is the second term in \eqref{bobR}.

	To consider $I(\bar{\mathbf{s}}_{A};\mathbf{y}_{E}|\hat{\mathbf{a}})$ in \eqref{eq:R_one}, we let $\hat{\mathbf{A}} = ivec(\hat{\mathbf{a}})$. Similar to the discussions leading to \eqref{legiR1} and \eqref{bobR}, one can verify that
\begin{equation}\label{eq:R_E_-}
  I(\bar{\mathbf{s}}_{A};\mathbf{y}_{E}|\hat{\mathbf{a}})\geq  K_{2}\mathcal{E}\{\log| \mathbf{I}_{N_E} +  \frac{P_{A}/N_{A}\hat{\mathbf{A}}\hat{\mathbf{A}}^{H}}{1 + P_{A}\sigma^{2}_{EA}} |\}
  \triangleq \mathcal{R}_{E}^{-}
\end{equation}
and
	\begin{equation}\label{eveR}
	\begin{aligned}
	&I(\bar{\mathbf{s}}_{A};\mathbf{y}_{E}|\hat{\mathbf{a}})\\
	&\leq \mathcal{R}_{E}^{-}
	+ N_{E}\log\bigg(\frac{(1 + P_{A}\sigma^{2}_{EA})^{K_{2}}}{\big(1 + (P_{A}\sigma^{2}_{EA}/N_{A})e_{A}\big)^{t_{A}}}\bigg)\\
	&\triangleq\mathcal{R}_{E}^{+}
	\end{aligned}
	\end{equation}

When $P_{A} = P_{B}  = P \rightarrow \infty$, we have  $\sigma_{EA}^{2} \rightarrow \frac{bN_{A}}{aN_{B} + bN_{A}}$, $\sigma_{B}^{2} \rightarrow 0$,  $\mathcal{E}\{\hat{a}_{i}\hat{a}_{i}^{*}\} \rightarrow \frac{aN_{B}}{aN_{B} + bN_{A}}$ and   $\mathcal{E}\{\hat{h}_{B,i}\hat{h}^{*}_{B,i}\} \rightarrow  1$. From \cite[Th.2]{Grant2002}, we know that $\mathcal{E}\{\log|\mathbf{I}_r+\frac{P}{t}\mathbf{X}\mathbf{X}^H|\}\to \min(r,t)\log P+o(\log P)$ as $P\to \infty$ where the entries of $\mathbf{X}\in \mathbb{C}^{r\times t}$ are i.i.d. $\mathcal{CN}(0,1)$. Therefore, from \eqref{legiR1} and \eqref{bobR},
	\begin{equation}\label{HighP_asym}
	\begin{aligned}
	\lim_{P\rightarrow\infty}\frac{\mathcal{R}_{B}^{-}}{\log P} = \lim_{P\rightarrow\infty}\frac{\mathcal{R}_{B}^{+}}{\log P} =  K_{2}\min\{N_{A},N_{B}\}
	\end{aligned}
	\end{equation}
And from \eqref{eq:R_E_-} and \eqref{eveR}, we have
	\begin{equation}\label{HighP_asym_e1}
	\lim_{P\rightarrow\infty}\frac{\mathcal{R}_{E}^{-}}{\log P} = 0
	\end{equation}
	and
	\begin{equation}\label{HighP_asym_e2}
	\lim_{P\rightarrow\infty}\frac{\mathcal{R}_{E}^{+}}{\log P} = \left\{
	\begin{aligned}
	& 0 ,&&K_2\leq N_{A}  \\
	&N_{E}(K_{2} - N_{A}),&&K_2>N_{A}
	\end{aligned}\right.
	\end{equation}
	Combining \eqref{HighP_asym}, \eqref{HighP_asym_e1} and \eqref{HighP_asym_e2} and using $	\mathcal{R}_{one}^{+}\triangleq\frac{1}{K_{2}}[\mathcal{R}_{B}^{+} - \mathcal{R}_{E}^{-}]^{+}$ and $	\mathcal{R}_{one}^{-}\triangleq\frac{1}{K_{2}}[\mathcal{R}_{B}^{-} - \mathcal{R}_{E}^{+}]^{+}$ (i.e., $\mathcal{R}_{one}^{-}\leq \mathcal{R}_{one}\leq \mathcal{R}_{one}^{+}$),  we have
	\begin{equation}\label{sdof_l}
	\begin{aligned}
	&\lim_{P\rightarrow \infty}\frac{\mathcal{R}_{one}^{-}}{\log P}\\
	& =
	\left\{
	\begin{aligned}
	& \min\{N_{A},N_{B}\},~K_2\leq N_{A}  \\
	&\bigg(\min\{N_{A},N_{B}\} - \frac{N_{E}}{K_2}(K_{2} - N_{A} )\bigg)^{+}, ~K_2>N_{A}
	\end{aligned}\right.
	\end{aligned}
	\end{equation}
	and
	\begin{equation}\label{sdof_u}
	\lim_{P\rightarrow \infty}\frac{\mathcal{R}_{one}^{+}}{\log P} =\min\{N_{A},N_{B}\}.
	\end{equation}
Note that $\lim_{P\rightarrow \infty}\frac{\mathcal{R}_{one}}{\log P}$ is called the secure degrees of freedom of the one-way information transmission.	
From \eqref{sdof_l} and \eqref{sdof_u}, we see that when $K_2\leq N_{A}$, we have $\lim_{P\rightarrow \infty}\frac{\mathcal{R}_{one}}{\log P} = \min\{N_{A},N_{B}\}$ which equals the degrees of freedom of the main channel capacity from Alice to Bob.  This supports and complements a conclusion from \cite{hua2018advanced} where the analysis did not use the complete statistical model of $\mathbf{H}$, $\mathbf{A}$ and $\mathbf{B}$. We also see from \eqref{sdof_l} that if $K_2>N_A$, the above lower bound on secure degrees of freedom decreases linearly as $N_E$ increases.

\subsubsection{Two-way information transmission}
Now we consider a two-way (full-duplex) communication in the second phase where the signals received by Alice, Bob and Eve in a coherence period are
\begin{equation}\label{ph2_two}
\begin{aligned}
&\mathbf{Y}_{A} = \mathbf{H}^{T}\mathbf{S}_{B} + \mathbf{N}_{A}\\
&\mathbf{Y}_{B} = \mathbf{H}\mathbf{S}_{A} + \mathbf{N}_{B}\\
&\mathbf{Y}_{E} = \sqrt{a}\mathbf{A}\mathbf{S}_{A} + \sqrt{b}\mathbf{B}\mathbf{S}_{B} + \mathbf{N}_{E}
\end{aligned}
\end{equation}
where $\mathbf{S}_{A} = [\mathbf{s}_{A}(1),\dots, \mathbf{s}_{A}(K_{2})]$ and $\mathbf{s}_{A}(t)\sim\mathcal{CN}(0, \frac{P_{A}}{N_{A}}\mathbf{I})$. Similarly $\mathbf{S}_{B} = [\mathbf{s}_{B}(1),\dots, \mathbf{s}_{B}(K_{2})]$ and $\mathbf{s}_{B}(t)\sim\mathcal{CN}(0, \frac{P_{B}}{N_{B}}\mathbf{I})$. Note that all information symbols from Alice and Bob are i.i.d..
The vectorized forms of \eqref{ph2_two} are
\begin{equation}\label{ph2_two_vec}
	\begin{aligned}
	&\mathbf{y}_{A} = (\mathbf{I}_{K_2}\otimes\mathbf{H}^{T})\bar{\mathbf{s}}_{B} + \mathbf{n}_{A}\\
	&\mathbf{y}_{B} = (\mathbf{I}_{K_2}\otimes\mathbf{H})\bar{\mathbf{s}}_{A} + \mathbf{n}_{B}\\
	&\mathbf{y}_{E} = \sqrt{a}(\mathbf{I}_{K_2}\otimes\mathbf{A})\bar{\mathbf{s}}_{A} + \sqrt{b}(\mathbf{I}_{K_2}\otimes\mathbf{B})\bar{\mathbf{s}}_{B} + \mathbf{n}_{E}
	\end{aligned}
\end{equation}
where both $\bar{\mathbf{s}}_{A}$ and $\bar{\mathbf{s}}_{B}$ are assumed to be independent of all channel parameters.
Conditional on the MMSE channel estimation in phase 1, an achievable secrecy rate in phase 2 by the two-way wiretap channel is (e.g., see \cite{Tekin2008}):
\begin{equation}\label{twoway}
\begin{aligned}
\mathcal{R}_{two} = &\frac{1}{K_{2}}\big(I(\bar{\mathbf{s}}_{B};\mathbf{y}_{A}|\hat{\mathbf{h}}_{A}) + I(\bar{\mathbf{s}}_{A};\mathbf{y}_{B}|\hat{\mathbf{h}}_{B})\\
&\quad  - I(\bar{\mathbf{s}}_{A},\bar{\mathbf{s}}_{B};\mathbf{y}_{E}|\hat{\mathbf{a}},\hat{\mathbf{b}})\big)^{+}
\end{aligned}
\end{equation}
The following analysis is similar to the previous section, for which we will only provide the key steps and results.

From \eqref{legiR1} and \eqref{bobR}, we already know a pair of lower and upper bounds on $I(\bar{\mathbf{s}}_{A};\mathbf{y}_{B}|\hat{\mathbf{h}}_{B})$. To show a similar pair of lower and upper bounds on $I(\bar{\mathbf{s}}_{B};\mathbf{y}_{A}|\hat{\mathbf{h}}_{A})$, we let
$\hat{\mathbf{H}}_{A} = ivec(\hat{\mathbf{h}}_{A})$. One can verify that
\begin{equation}\label{legiR2}
\begin{aligned}
&I(\bar{\mathbf{s}}_{B};\mathbf{y}_{A}|\hat{\mathbf{h}}_{A})\\
& \geq K_{2}\mathcal{E}\{\log|\mathbf{I}_{N_A} + \frac{P_{B}/N_{B}}{1 + \frac{\sigma^{2}P_{B}}{1 +\sigma^{2}T_{1}P_{B}/N_{B} }}\hat{\mathbf{H}}^T_{A}\hat{\mathbf{H}}^{*}_{A}|\}
\triangleq\mathcal{R}_{A}^{-}
\end{aligned}
\end{equation}
and
\begin{equation}
\begin{aligned}
I(\bar{\mathbf{s}}_{B};\mathbf{y}_{A}|\hat{\mathbf{h}}_{A}) &\leq  \mathcal{R}_{A}^{-} + N_{A}\log\bigg(\frac{(1 + \frac{P_{B}}{1 + K_{1}P_{B}/N_{B}})^{K_{2}}}{\big(1 + \frac{P_{B}/N_{B}}{1 + K_{1}P_{B}/N_{B}}e_{B} \big)^{t_{B}}}\bigg)\\
&\triangleq\mathcal{R}_{A}^{+}
\end{aligned}
\end{equation}
where $e_{B} = \exp(\frac{1}{t_{B}}\sum_{j =1}^{t_{B}}\sum_{k=1}^{r_{B} - j}\frac{1}{k} - \gamma)$, $t_{B} = \min\{N_{B},K_{2}\}$ and $r_{B} = \max\{N_{B},K_{2}\}$.
For $I(\bar{\mathbf{s}}_{A},\bar{\mathbf{s}}_{B};\mathbf{y}_{E}|\hat{\mathbf{a}},\hat{\mathbf{b}})$, we use $\hat{\mathbf{B}} = ivec(\hat{\mathbf{b}})$ (similar to $\hat{\mathbf{A}}$). One can verify that $\mathbf{K}_{\mathbf{y}_{E}|\hat{\mathbf{a}},\hat{\mathbf{b}}} = \frac{P_{A}}{N_{A}}(\mathbf{I}_{K_2}\otimes\hat{\mathbf{A}}\hat{\mathbf{A}}^{H}) + \frac{P_{B}}{N_{B}}(\mathbf{I}_{K_2}\otimes\hat{\mathbf{B}}\hat{\mathbf{B}}^{H}) + \mathbf{K}_{EA} + \mathbf{K}_{EB} + \mathbf{I}_{N_EK_2}$ where $\mathbf{K}_{EA} = \mathcal{E}\{(\mathbf{I}_{K_2}\otimes\Delta{\mathbf{A}})\bar{\mathbf{s}}_{A}
\bar{\mathbf{s}}_{A}^{H}(\mathbf{I}_{K_2}\otimes\Delta{\mathbf{A}})^{H}\} = \sigma_{EA}^{2}P_{A}\mathbf{I}_{N_EK_2}$ and  $\mathbf{K}_{EB} = \mathcal{E}\{(\mathbf{I}_{K_2}\otimes\Delta{\mathbf{B}})\bar{\mathbf{s}}_{B}
\bar{\mathbf{s}}_{B}^{H}(\mathbf{I}_{K_2}\otimes\Delta{\mathbf{B}})^{H}\} = \sigma_{EB}^{2}P_{B}\mathbf{I}_{N_EK_2}$.
Also note that $\mathbf{y}_{E} = (\mathbf{S}_{A}^{T}\otimes\mathbf{I}_{N_{E}})\mathbf{h}_{EA} + (\mathbf{S}_{B}^{T}\otimes\mathbf{I}_{N_{E}})\mathbf{h}_{EB} + \mathbf{n}_{E}$. Then,
\begin{equation}\label{eveR_two}
\begin{aligned}
&I(\bar{\mathbf{s}}_{A},\bar{\mathbf{s}}_{B};\mathbf{y}_{E}|\hat{\mathbf{a}},\hat{\mathbf{b}})\\
&= h(\mathbf{y}_{E}|\hat{\mathbf{a}},\hat{\mathbf{b}}) - h(\mathbf{y}_{E}|\hat{\mathbf{a}},\hat{\mathbf{b}},\bar{\mathbf{s}}_{A},\bar{\mathbf{s}}_{B})\\
&\leq \mathcal{E}\{\log[(\pi e)^{K_{2}N_{E}}|\mathbf{K}_{\mathbf{y}_{E}|\hat{\mathbf{a}},\hat{\mathbf{b}}}|]\}-h(\mathbf{y}_{E}|\hat{\mathbf{a}},\hat{\mathbf{b}},\bar{\mathbf{s}}_{A},\bar{\mathbf{s}}_{B})\\
& = \mathcal{E}\{\log|\mathbf{K}_{\mathbf{y}_{E}|\hat{\mathbf{a}},\hat{\mathbf{b}}}|\}- \mathcal{E}\{\log|\sigma^{2}_{EA}(\mathbf{S}_{A}^{T}\mathbf{S}_{A}^{*}\otimes \mathbf{I}_{N_E})\\
&\quad  + \sigma^{2}_{EB}(\mathbf{S}_{B}^{T}\mathbf{S}_{B}^{*}\otimes \mathbf{I}_{N_E})  +\mathbf{I}_{N_EK_2}|\}\\
& =  K_{2}\mathcal{E}\{\log|\frac{P_{A}}{N_{A}}\hat{\mathbf{A}}\hat{\mathbf{A}}^{H} +  \frac{P_{B}}{N_{B}}\hat{\mathbf{B}}\hat{\mathbf{B}}^{H} + (1 + P_{A}\sigma^{2}_{EA} \\
&\quad  + P_{B}\sigma^{2}_{EB})\mathbf{I}_{N_E}|\}\\
&\quad - N_{E}\mathcal{E}\{\log|\sigma^{2}_{EA}\mathbf{S}_{A}^{T}\mathbf{S}_{A}^{*} + \sigma^{2}_{EB}\mathbf{S}_{B}^{T}\mathbf{S}_{B}^{*}  +\mathbf{I}_{K_2}|\}
\end{aligned}
\end{equation}

Define $\mathbf{S}_{AB} = [\check{\mathbf{S}}_{A}^{T}, \check{\mathbf{S}}_{B}^{T}] \in \mathbb{C}^{K_{2} \times (N_{A} + N_{B})}$ where $\mathbf{S}_{A} = \frac{P_{A}}{N_{A}}\check{\mathbf{S}}_{A}$ and $\mathbf{S}_{B} = \frac{P_{B}}{N_{B}}\check{\mathbf{S}}_{B}$. Define $\mathbf{T} = diag\{\sigma_{EA}^{2}\frac{P_{A}}{N_{A}}\mathbf{I}_{N_{A}}, \sigma_{EB}^{2}\frac{P_{B}}{N_{B}}\mathbf{I}_{N_{B}}\}$.
Then we can rewrite the last term from  \eqref{eveR_two} as $\mathcal{E}\{\log|\sigma^{2}_{EA}\mathbf{S}_{A}^{T}\mathbf{S}_{A}^{*} + \sigma^{2}_{EB}\mathbf{S}_{B}^{T}\mathbf{S}_{B}^{*}  +\mathbf{I}_{K_2}|\} = \mathcal{E}\{\log|\mathbf{I}_{K_2} + \mathbf{S}_{AB}\mathbf{T}\mathbf{S}_{AB}^{H}|\}$.

For $K_2<N_{A} + N_{B} $, we have
\begin{equation}\label{l1}
	\begin{aligned}
	&\mathcal{E}\{\log|\mathbf{I}_{K_2} + \mathbf{S}_{AB}\mathbf{T}\mathbf{S}_{AB}^{H}|\}\\
	&\geq K_{2}\mathcal{E}\{\log(1 + |\mathbf{S}_{AB}\mathbf{T}\mathbf{S}_{AB}^{H}|^{\frac{1}{K_{2}}})\}\\
	& = K_{2}\mathcal{E}\big\{\log\big(1 + \exp\big(\frac{1}{K_{2}}\ln|\mathbf{S}_{AB}\mathbf{T}\mathbf{S}_{AB}^{H}|\big)\big)\big\}\\
	&\geq K_{2}\mathcal{E}\big\{\log\big(1 +
	\exp\big(\frac{1}{K_{2}}\ln\sigma^{2K_{2}}_{min}|\mathbf{S}_{AB}\mathbf{S}_{AB}^{H}|\big)\big)\big\}\\
	& \geq K_{2}\log\big(1 + \sigma^{2}_{min}e_{E1}\big)
	\end{aligned}
\end{equation}
where $e_{E1} = \exp(\frac{1}{K_{2}}\sum_{j =1}^{K_{2}}\sum_{k=1}^{N_{A}+N_{B} - j}\frac{1}{k} - \gamma)$. The second inequality in \eqref{l1} is from the fact (e.g., see \cite[Th. 3]{Jin2007}) that $|\mathbf{S}_{AB}\mathbf{T}\mathbf{S}_{AB}^{H}| \geq \sigma^{2K_{2}}_{min}|\mathbf{S}_{AB}\mathbf{S}_{AB}^{H}|$ where $\sigma^{2}_{min} = \min\{\sigma_{EA}^{2}\frac{P_{A}}{N_{A}}, \sigma_{EB}^{2}\frac{P_{B}}{N_{B}}\}$.
Similarly, for $K_2\geq N_{A} + N_{B}$, we have
\begin{equation}\label{l2}
	\begin{aligned}
	&\mathcal{E}\{\log|\mathbf{I} + \mathbf{S}_{AB}\mathbf{T}\mathbf{S}_{AB}^{H}|\}\\
	& = \mathcal{E}\{\log|\mathbf{I} + \mathbf{T}\mathbf{S}_{AB}^{H}\mathbf{S}_{AB}|\}\\
	&\geq (N_{A} + N_{B})\mathcal{E}\big\{\log\big(1 \\
&\quad + |\mathbf{T}|^{\frac{1}{N_{A} + N_{B}}}exp\big(\frac{1}{N_{A} + N_{B}}\ln|\mathbf{S}_{AB}^{H}\mathbf{S}_{AB}|\big)\big)\big\}\\
	&\geq (N_{A} + N_{B})\log\big(1 + |\mathbf{T}|^{\frac{1}{N_{A} + N_{B}}}e_{E2}\big)
	\end{aligned}
\end{equation}
where $e_{E2} = \exp(\frac{1}{N_{A} + N_{B}}\sum_{j =1}^{N_{A} + N_{B}}\sum_{k=1}^{K_{2} - j}\frac{1}{k} - \gamma)$
Therefore, using \eqref{l1} and \eqref{l2}, we have from \eqref{eveR_two} that
\begin{equation}\label{eveR_up}
\begin{aligned}
&I(\bar{\mathbf{s}}_{A},\bar{\mathbf{s}}_{B};\mathbf{y}_{E}|\hat{\mathbf{a}},\hat{\mathbf{b}})\\
&\leq K_{2}\mathcal{E}\{\log|\frac{\frac{P_{A}}{N_{A}}\hat{\mathbf{A}}\hat{\mathbf{A}}^{H} +  \frac{P_{B}}{N_{B}}\hat{\mathbf{B}}\hat{\mathbf{B}}^{H}}{1 + P_{A}\sigma^{2}_{EA} + P_{B}\sigma^{2}_{EB}} + \mathbf{I}|\} \\
&{\small + \left\{
\begin{aligned}
&K_{2}N_{E}\log\bigg(\frac{1 + P_{A}\sigma^{2}_{EA} + P_{B}\sigma^{2}_{EB}}{1 + \sigma^{2}_{min}e_{E1}}\bigg),~K_2\leq N_{A} + N_{B}\\
&N_{E}\log\bigg(\frac{(1 + P_{A}\sigma^{2}_{EA} + P_{B}\sigma^{2}_{EB})^{K_{2}}}{\big(1 + |\mathbf{T}|^{\frac{1}{N_{A} + N_{B}}}e_{E2}\big)^{N_{A} + N_{B}}}\bigg), ~K_2>N_{A} + N_{B}
\end{aligned}  \right.}\\
&\triangleq\mathcal{R}_{E,t}^{+}
\end{aligned}
\end{equation}

One can also verify $I(\bar{\mathbf{s}}_{A},\bar{\mathbf{s}}_{B};\mathbf{y}_{E}|\hat{\mathbf{a}},\hat{\mathbf{b}}) \geq K_{2}\mathcal{E}\{\log|\frac{\frac{P_{A}}{N_{A}}\hat{\mathbf{A}}\hat{\mathbf{A}}^{H} +  \frac{P_{B}}{N_{B}}\hat{\mathbf{B}}\hat{\mathbf{B}}^{H}}{1 + P_{A}\sigma^{2}_{EA} + P_{B}\sigma^{2}_{EB}} + \mathbf{I}|\}\triangleq  \mathcal{R}_{E,t}^{-} $ which is the first term in \eqref{eveR_up}.
When $P_{A} = P_{B}  = P \rightarrow \infty$, we have  $\sigma_{EA}^{2} \rightarrow \frac{bN_{A}}{aN_{B} + bN_{A}}$, $\sigma_{EB}^{2} \rightarrow \frac{aN_{B}}{aN_{B} + bN_{A}}$,  $\sigma_{A}^{2} \rightarrow 0$, $\sigma_{B}^{2} \rightarrow 0$,   $\mathcal{E}\{\hat{a}_{i}\hat{a}_{i}^{*}\} \rightarrow \frac{aN_{B}}{aN_{B} + bN_{A}}$,  $\mathcal{E}\{\hat{b}_{i}\hat{b}_{i}^{*}\} \rightarrow \frac{bN_{A}}{aN_{B} + bN_{A}}$, $\mathcal{E}\{\hat{h}_{A,i}\hat{h}^{*}_{A,i}\} \rightarrow  1$,  $\mathcal{E}\{\hat{h}_{B,i}\hat{h}^{*}_{B,i}\} \rightarrow  1$, $\sigma^{2}_{min} = P\min\{\frac{\sigma_{EA}^{2}}{N_{A}}, \frac{\sigma_{EB}^{2}}{N_{B}}\}$ and $|\mathbf{T}|^{\frac{1}{N_{A} + N_{B}}} = P((\frac{\sigma_{EA}^{2}}{N_{A}})^{N_{A}}(\frac{\sigma_{EB}^{2}}{N_{B}})^{N_{B}})^{1/(N_{A}+N_{B})}$.
Then, similar to \eqref{HighP_asym}, we have
\begin{equation}\label{HighP_asym2}
	\lim_{P\rightarrow\infty}\frac{\mathcal{R}_{A}^{-}}{\log P} = \lim_{P\rightarrow\infty}\frac{\mathcal{R}_{A}^{+}}{\log P} =  K_{2}\min\{N_{A},N_{B}\}
\end{equation}
One can also verify that
\begin{equation}\label{HighP_asym_e3}
\lim_{P\rightarrow\infty}\frac{\mathcal{R}_{E,t}^{-}}{\log P} = 0
\end{equation}
and
\begin{equation}\label{HighP_asym_e4}
\lim_{P\rightarrow\infty}\frac{\mathcal{R}_{E,t}^{+}}{\log P} = \left\{
\begin{aligned}
& 0 ,&&K_2\leq N_{A}+N_{B}  \\
&N_{E}(K_{2} - N_{A}-N_{B}),&&K_2>N_{A}+N_{B}
\end{aligned}\right.
\end{equation}
Now applying \eqref{HighP_asym}, \eqref{HighP_asym2}, \eqref{HighP_asym_e3} and \eqref{HighP_asym_e4}, and using $	\mathcal{R}_{two}^{+}\triangleq\frac{1}{K_{2}}[\mathcal{R}_{A}^{+}+\mathcal{R}_{B}^{+} - \mathcal{R}_{E,t}^{-}]^{+}$ and $	\mathcal{R}_{two}^{-}\triangleq\frac{1}{K_{2}}[\mathcal{R}_{A}^{-}+\mathcal{R}_{B}^{-} - \mathcal{R}_{E,t}^{+}]^{+}$ as upper and lower bounds on $\mathcal{R}_{two}$, we have
\begin{equation}
\begin{aligned}
&\lim_{P\rightarrow \infty}\frac{\mathcal{R}_{two}^{-}}{\log P}\\
& =
\left\{
\begin{aligned}
& 2\min\{N_{A},N_{B}\},~K_2\leq N_{A} + N_{B} \\
&\bigg(2\min\{N_{A},N_{B}\} - \frac{N_{E}}{K_2}(K_{2} - N_{A} - N_{B})\bigg)^{+}\\
&\qquad\qquad\qquad\qquad,~K_2>N_{A} + N_{B}
\end{aligned}\right.
\end{aligned}
\end{equation}
and
\begin{equation}
\lim_{P\rightarrow \infty}\frac{\mathcal{R}_{two}^{+}}{\log P} =2\min\{N_{A},N_{B}\}
\end{equation}
We see that if $K_2\leq N_{A} + N_{B}$, $\lim_{P\rightarrow \infty}\frac{\mathcal{R}_{two}}{\log P} = 2\min\{N_{A},N_{B}\}$ which equals  the degrees of freedom of the full-duplex channel between Alice and Bob. And if $K_2>N_{A} + N_{B}$, the above lower bound on $\lim_{P\rightarrow \infty}\frac{\mathcal{R}_{two}}{\log P} $ decreases linearly as $N_E$ increases. We see an advantage of two-way information transmission over one-way information transmission.

\subsection{Eve uses blind detection with zero knowledge of its CSI}
\label{sec:long_anece}
Now we reconsider the case of one-way information transmission from Alice to Bob in the second phase but assume that Eve performs a blind detection of the information transmitted from Alice. For the blind detection shown next, we also assume that $K_2> N_A$ and Eve's knowledge of its CSI matrix $\sqrt{a}\mathbf{A}\in \mathbb{C}^{N_E\times K_2}$ is zero. (The two-way information transmission between Alice and Bob in either half-duplex or ideal full-duplex can be treated similarly. For the case of $K_2\leq N_A$, Eve cannot receive any information from the users due to its unknown CSI.)

The signal received by Eve during information transmission from Alice over $K_2$ sampling intervals is
\begin{equation}
    \mathbf{Y}_E = \sqrt{a}\mathbf{A}\mathbf{S}_A+\mathbf{N_E}
\end{equation}
where the elements in $\mathbf{S}_A\in \mathbb{C}^{N_A\times K_2}$ are assumed to be independently chosen from a known constellation $\mathbb{S}_N$ with size $N$.
Assume that Eve performs the blind detection as follows:
\begin{equation}
    \left(\hat{\mathbf{S}}, \hat{\mathbf{A}}\right) = \argmin_{\mathbf{S} \in \mathbb{S}_N^{N_A\times K_2}, \sqrt{a}\mathbf{A}\in \mathbb{C}^{N_E\times K_2}} \|\mathbf{Y}_E - \sqrt{a}\mathbf{A}\mathbf{S}\|^2_F.
\end{equation}
Given any $\mathbf{S}$, the optimal $\sqrt{a}\mathbf{A}$ is  $\mathbf{Y}_E \mathbf{S}^H\left (\mathbf{S}\mathbf{S}^H\right)^{-1}$. Then, the above problem reduces to the following (an issue of uniqueness will be addressed later)
\begin{equation}\label{eq:P0}
    \hat{\mathbf{S}} = \argmin_{\mathbf{S} \in \mathbb{S}_N^{N_A\times K_2}} \|\mathbf{Y}_E - \mathbf{Y}_E \mathbf{S}^H\left (\mathbf{S}\mathbf{S}^H\right)^{-1}\mathbf{S}\|^2_F,
\end{equation}
or equivalently
$
    \hat{\mathbf{S}} = \argmax_{\mathbf{S} \in \mathbb{S}_N^{N_A\times K_2}} f\left (\mathbf{s}\right),
$
where $f\left (\mathbf{s}\right) = \textrm{Tr}\left ( \mathbf{S}^H\left (\mathbf{S}\mathbf{S}^H\right)^{-1}\mathbf{S}\mathbf{Z}\right)$, $\mathbf{s} = \textrm{vec}\left (\mathbf{S}\right)$ and $\mathbf{Z} = \mathbf{Y}_E^H\mathbf{Y}_E$. The above problem is computationally expensive. But we assume that Eve is able to afford it.

Assume that  the solution $\hat{\mathbf{S}}$ of the above problem  is so close to the actual information matrix  $\mathbf{S}_0$ that $f\left (\mathbf{s}\right)$ can be replaced by its 2nd-order Taylor's series expansion (which is conservative for Alice and Bob or equivalently optimistic for Eve). Then $\hat{\mathbf{s}} = \textrm{vec}\left (\hat{\mathbf{S}}\right)$ has the following properties
\begin{equation}
\label{eq:nabla_f}
    \nabla_{\mathbf{s}}\tilde{f}\left (\mathbf{s}\right)|_{\mathbf{s} = \hat{\mathbf{s}}} = \mathbf{0},
\end{equation}
\begin{equation}
\label{eq:nabla_f*}
    \nabla_{\mathbf{s}^*}\tilde{f}\left (\mathbf{s}\right) \left|_{\mathbf{s}= \hat{\mathbf{s}}}\right.= \mathbf{0} ,
\end{equation}
where $\tilde{f}\left (\mathbf{s}\right)$ is the second-order Taylor series expansion \cite{kreutz2009complex} of $f\left (\mathbf{s}\right)$ around $\mathbf{s}_0 = \textrm{vec}\left (\mathbf{S}_0\right)$, i.e.,
\begin{align}
    \tilde{f}\left (\mathbf{s}\right)&=\notag\\& f\left (\mathbf{s}_0\right)+ \nabla_{\mathbf{s}}^Tf\left (\mathbf{s}\right)|_{\mathbf{s}=\mathbf{s}_0}\left (\mathbf{s}-\mathbf{s}_0\right)+\nabla_{\mathbf{s}^*}^Tf\left (\mathbf{s}\right)|_{\mathbf{s}=\mathbf{s}_0}\left (\mathbf{s}-\mathbf{s}_0\right)^*\notag \\& +\frac{1}{2}\left [ \left (\mathbf{s}-\mathbf{s}_0\right)^H \mathbf{H}_{ss}\left (\mathbf{s}-\mathbf{s}_0\right)+\left (\mathbf{s}-\mathbf{s}_0\right)^H \mathbf{H}_{s^*s}\left (\mathbf{s}-\mathbf{s}_0\right)^*\right.\notag\\& \left. +\left (\mathbf{s}-\mathbf{s}_0\right)^T \mathbf{H}_{ss^*}\left (\mathbf{s}-\mathbf{s}_0\right)+\left (\mathbf{s}-\mathbf{s}_0\right)^T \mathbf{H}_{s^*s^*}\left (\mathbf{s}-\mathbf{s}_0\right)^*\right ].
\end{align}
which involves the Hessian matrices: $\mathbf{H}_{ss} = \frac{\partial}{\partial \mathbf{s}}\left (\nabla_{\mathbf{s}^*} f\right)\left|_{\mathbf{s}= {\mathbf{s}}_0}\right.$, $\mathbf{H}_{s^*s} = \frac{\partial}{\partial \mathbf{s}^*}\left (\nabla_{\mathbf{s}^*} f\right)\left|_{\mathbf{s}= {\mathbf{s}}_0}\right.$, $\mathbf{H}_{ss^*} = \frac{\partial}{\partial \mathbf{s}}\left (\nabla_{\mathbf{s}} f\right)\left|_{\mathbf{s}= {\mathbf{s}}_0}\right.$, $\mathbf{H}_{s^*s^*} = \frac{\partial}{\partial \mathbf{s}^*}\left (\nabla_{\mathbf{s}} f\right)\left|_{\mathbf{s}= {\mathbf{s}}_0}\right.$. Subject to uniqueness of solution, solving   (\ref{eq:nabla_f}) and   (\ref{eq:nabla_f*}) results in the following \cite{kreutz2009complex}
\begin{equation}
    \Hat{\mathbf{s}}-\mathbf{s}_0 = \left (\mathbf{H}_{ss}-\mathbf{H}_{s^*s}\mathbf{H}_{ss}^{-T}\mathbf{H}_{s^*s}^H\right)^{-1}\left (\mathbf{H}_{s^*s}\mathbf{H}_{ss}^{-T}\nabla_{\mathbf{s}} f-\nabla_{\mathbf{s}^*} f\right).
\end{equation}
Furthermore,
\begin{equation}
    \nabla_{\mathbf{s}^*} f = \left (\mathbf{Z}\left (\mathbf{I}-\mathbf{S}^H\left (\mathbf{S}\mathbf{S}^H\right)^{-1}\mathbf{S}\right)\right)^T\otimes \left (\mathbf{S}\mathbf{S}^H\right)^{-1}\mathbf{s},
\end{equation}
\begin{equation}
    \nabla_{\mathbf{s}} f = \left (\nabla_{\mathbf{s}} f\right)^*,
\end{equation}
\begin{align}
    &\mathbf{H}_{ss}  =\notag \\&\left[\mathbf{Z}-\mathbf{Z}\mathbf{S}^H\left (\mathbf{S}\mathbf{S}^H\right)^{-1}\mathbf{S} - \mathbf{S}^H\left (\mathbf{S}\mathbf{S}^H\right)^{-1}\mathbf{S}\mathbf{Z}\right. \notag\\& \left. + \mathbf{S}^H\left (\mathbf{S}\mathbf{S}^H\right)^{-1}\mathbf{S}\mathbf{Z}\mathbf{S}^H\left (\mathbf{S}\mathbf{S}^H\right)^{-1}\mathbf{S} \right]^T\otimes \left (\mathbf{S}\mathbf{S}^H\right)^{-1} \notag \\ &+ \left (\mathbf{S}^H\left (\mathbf{S}\mathbf{S}^H\right)^{-1}\mathbf{S}-\mathbf{I}\right)^T \otimes \mathbf{S}^H\left (\mathbf{S}\mathbf{S}^H\right)^{-1}\mathbf{S}\mathbf{Z}\mathbf{S}^H\left (\mathbf{S}\mathbf{S}^H\right)^{-1}\mathbf{S},
\end{align}
\begin{align}
    &\mathbf{H}_{s^*s}  =\notag \\& \left [\left (\left (\left (\mathbf{S}\mathbf{S}^H\right)^{-1}\mathbf{S}\mathbf{Z}\right)\left (\mathbf{S}^H\left (\mathbf{S}\mathbf{S}^H\right)^{-1}\mathbf{S}-\mathbf{I}\right)\right)^T \otimes \left (\mathbf{S}\mathbf{S}^H\right)^{-1}\mathbf{S}+ \right. \notag \\ & \left.
    \left [\left (\mathbf{S}\mathbf{S}^H\right)^{-1}\mathbf{S}\right]^T \otimes \left (\left (\mathbf{S}\mathbf{S}^H\right)^{-1}\mathbf{S}\mathbf{Z}\right)\left (\mathbf{S}^H\left (\mathbf{S}\mathbf{S}^H\right)^{-1}\mathbf{S}-\mathbf{I}\right)\right ]\mathbf{\Pi}
\end{align}
where $\mathbf{\Pi}$ is a permutation matrix with
\begin{equation}
    \mathbf{\Pi}_{i,j} =\begin{cases} 1 & j = \left (\left (i-1\right)_{\bmod N_A}\right)K_2 + \left \lfloor \left (i-1\right)/N_A \right \rfloor \\ 0 & else \end{cases}
\end{equation}
where $a_{\bmod b}$ denotes the remainder of the division of $a$ by $b$. For more details about complex derivatives, please refer to \cite{kreutz2009complex}.

Because of the blind nature, $\mathbf{H}_{ss}$ is always rank deficient by $N_A^2$. To remove the ambiguity, we can treat the first $N_A$ of the transmitted vectors from Alice as known, which is equivalent to removing  $N_A^2$ corresponding rows and $N_A^2$ corresponding columns from each of $\mathbf{H}_{ss}$ and $\mathbf{H}_{s^*s}$, and  removing $N_A^2$ corresponding elements from each of $ \nabla_{\mathbf{s}} f$ and $ \nabla_{\mathbf{s}^*} f$. This results in $\bar{\mathbf{H}}_{ss}$, $\bar{\mathbf{H}}_{s^*s}$, $ \bar{\nabla}_{\mathbf{s}} f$ and $ \bar{\nabla}_{\mathbf{s}^*} f$, respectively. Hence the MSE matrix $\bar{\mathbf{M}}$ of the remaining unknown parameters can be formed as
\begin{equation}\label{eq:barM}
    \bar{\mathbf{M}} = \E \left [ \left (\overline{\hat{\mathbf{s}}-\mathbf{s}_0}\right)\left (\overline{\hat{\mathbf{s}}-\mathbf{s}_0}\right)^H\right ],
\end{equation}
where $\overline{\hat{\mathbf{s}}-\mathbf{s}_0}$ is the approximation of errors in the vector of all $N_A(K_2-N_A)$ remaining symbols and
\begin{equation}\label{eq:s-s}
    \overline{\Hat{\mathbf{s}}-\mathbf{s}_0} = \left (\bar{\mathbf{H}}_{ss}-\bar{\mathbf{H}}_{s^*s}\bar{\mathbf{H}}_{ss}^{-T}\bar{\mathbf{H}}_{s^*s}^H\right)^{-1}\left (\bar{\mathbf{H}}_{s^*s}\bar{\mathbf{H}}_{ss}^{-T}\bar{\nabla}_{\mathbf{s}} f-\bar{\nabla}_{\mathbf{s}^*} f\right)
\end{equation}
Finally, for $K_2>N_A$, Eve's effective rate (with the information in the first $N_A$ vectors of $\mathbf{s}(k)$ removed) can be approximated as
\begin{equation}\label{eq:RAE2}
    R_{AE}^{(2)} = \frac{1}{K_2}\left (\log |\bar{\mathbf{Q}}| - \log |\bar{\mathbf{M}}|\right),
\end{equation}
where $\bar{\mathbf{Q}}$ is the covariance matrix of the vector of all remaining symbols.

To evaluate $R_{AE}^{(2)}$, one has to specify the actual constellation $\mathbb{S}_N$ of each symbol in $\mathbf{S}$, compute $\overline{\Hat{\mathbf{s}}-\mathbf{s}_0}$ for each actual realization of $\mathbf{S}_0$ according to \eqref{eq:s-s}, and then obtain a sample averaged version of $\bar{\mathbf{M}}$ in \eqref{eq:barM}. Each of the realizations of $\mathbf{S}_0$ should be coupled with an independent realization of the channel matrix $\mathbf{A}$ and the noise matrix $\mathbf{N}_E$. With the final sample-averaged versions of $\bar{\mathbf{Q}}$ and $\bar{\mathbf{M}}$, $R_{AE}^{(2)}$ in \eqref{eq:RAE2} can be obtained.

For the next two plots, we assume that $\mathbb{S}_N$ is 4-QAM\footnote{For higher order constellations, the simulation became too slow and consuming.}, 100 random realizations of $\mathbf{S}_0$, $\mathbf{A}$ and $\mathbf{N}_E$ are used in computing $R_{AE}^{(2)}$. Also, during information transmission from Alice to Bob, $P_A=30dB$ (and $P_B=0$). In this case, due to high power, we expect the Taylor's series expansion applied in our derivation is accurate.

Fig. \ref{fig:conv_anace_blind} shows $R_{AE}^{(2)}$ versus $K_2/N_A$ where $N_A=N_B=4$ and $N_E=8$. We see that only when $K_2$ becomes much larger than $N_A$, $R_{AE}^{(2)}$ approaches $R_{AE}$. Note that $R_{AE}^{(2)}$ is based on unknown CSI at Eve and blind detection at Eve while $R_{AE}$ is based on the assumption that Eve knows its CSI perfectly.

\begin{figure}[b!]
    \centering
    \includegraphics[width=1\linewidth]{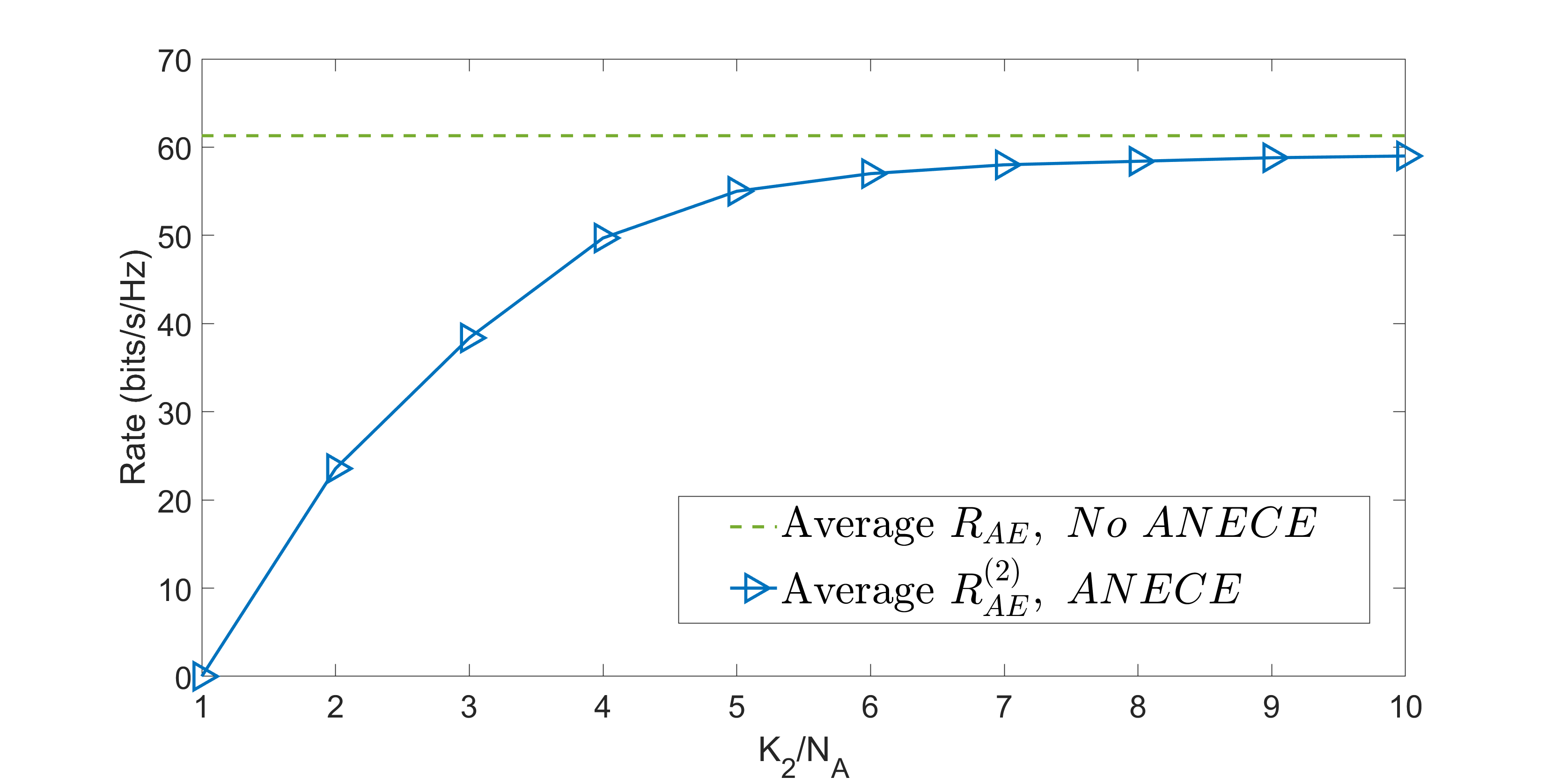}
    \caption{Eve's rates vs $K_2/N_A$ for known or unknown CSI at Eve. With ANECE, Eve does not know its CSI. Otherwise, Eve does.}
    \label{fig:conv_anace_blind}
\end{figure}

Fig. \ref{fig:comp_anace_blind} shows the averaged secret rate  $\bar R_S=(\mathcal{E}[R_{AB}-R_{AE}^{(2)}])^+$ versus $N_E$ where $N_A=N_B=4$. (The curves in this figure were zoomed in for the range of $N_E$ from $4$ to $20$. The actually computed points were at $N_E=4, 8, 16,32$.) In this case, $(\mathcal{E}[R_{AB}-R_{AE}])^+$ is zero for all values of $N_E$. But when Eve is blind to its CSI (caused by ANECE), the secrecy rates become substantial. In this case, we also see that for given $K_2>N_A$ the secrecy rate decreases as the number of antennas on Eve increases.

The above results in this subsection complement the analytical insights shown in the previous subsection (e.g., see \eqref{sdof_l}). Due to different assumptions, we cannot make a precise comparison between \eqref{sdof_l} and Figs. \ref{fig:conv_anace_blind} and \ref{fig:comp_anace_blind} while the general trends predicted in both cases are somewhat consistent. An additional discussion of the blind detection where Eve uses a partial knowledge of its CSI from phase 1 is shown in Appendix \ref{app2}.

\begin{figure}[bt!]
    \centering
    \includegraphics[width=1\linewidth]{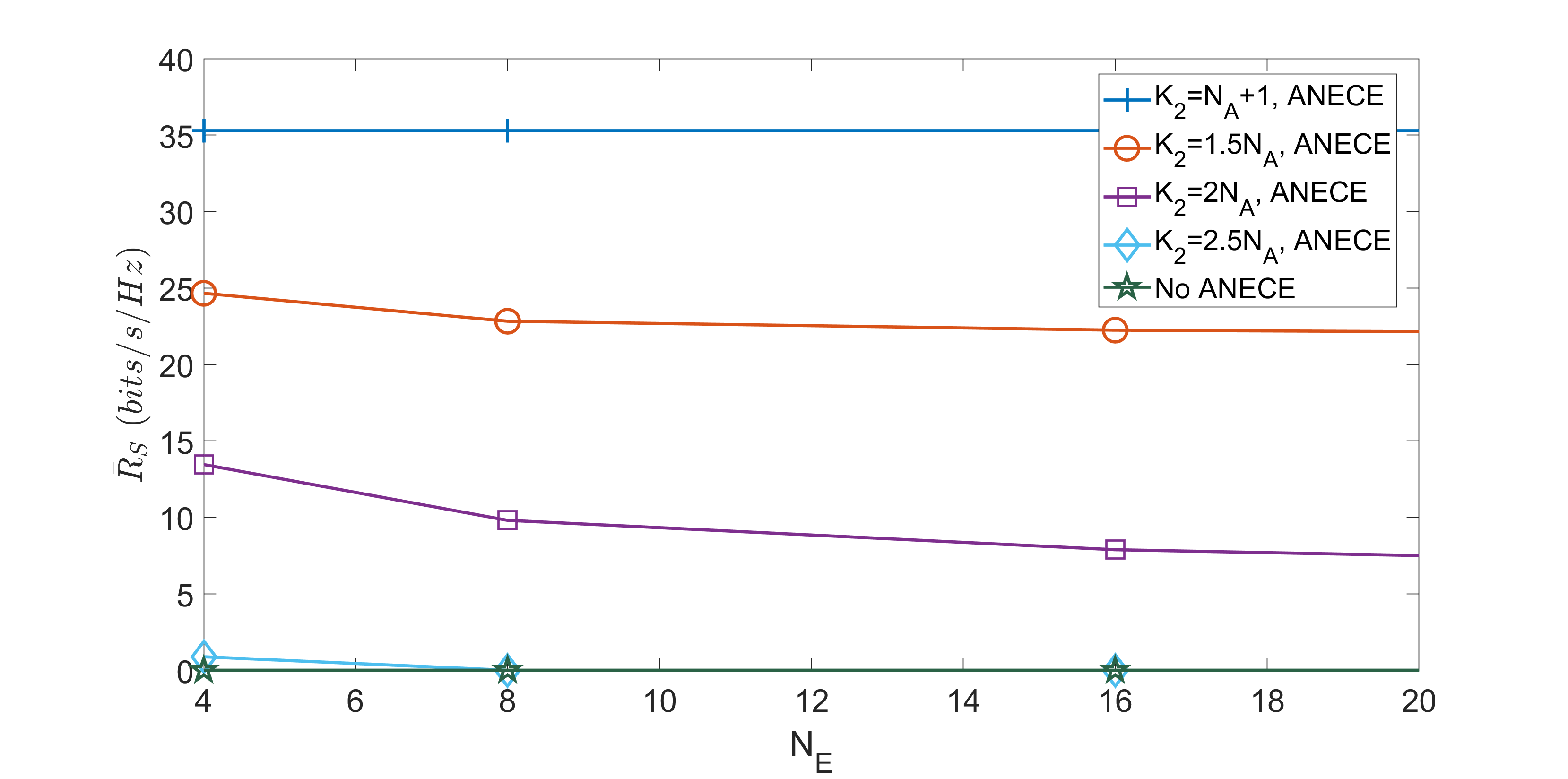}
    \caption{$\bar R_S$ vs. $N_E$ for known or unknown CSI at Eve. With ANECE, Eve does not know its CSI. Otherwise, Eve does.}
    \label{fig:comp_anace_blind}
\end{figure}

\section{Conclusion}
\label{sec:conclusion}
In this paper, we have investigated the secrecy performance of a full-duplex MIMOME network in some important scenarios. In the first part of this paper, we studied how to optimize the jamming powers from both Alice and Bob when Eve's CSI is unknown to Alice and Bob but Eve knows all CSI. To handle Eve's CSI being unknown to Alice and Bob, we focused on Eve at the most harmful location and adopted the large matrix theory that yields a hardened secret rate for any large number of antennas on Eve. With the optimized powers, we revealed a significant improvement in terms of the maximum tolerable number of antennas on Eve.  In the second part of this paper, we analyzed the full-duplex MIMOME network subject to the application of anti-eavesdropping channel estimation (ANECE) in a two-phase scheme. Assuming that a statistical model of CSI anywhere is known everywhere, we derived lower and upper bounds on the secure degrees of freedom of  the network, which reveal clearly how the number of antennas on Eve affect these bounds. In particular, for $1\leq K_2\leq N_A$ in one-way information transmission or $1\leq K_2\leq N_A+N_B$ in two-way information transmission, the lower and upper bounds coincide and equal to those of the channel capacity between Alice and Bob. Furthermore, assuming that Eve does not have any prior knowledge of its CSI but uses blind detection in phase 2 of the two-phase scheme, we provided and illustrated an approximate secrecy rate for $K_2>N_A$ in one-way information transmission. But the exact secrecy rate of the full-duplex MIMOME network with ANECE for $K_2$ larger than the total number of transmitting antennas still remains elusive. Nevertheless, the contributions shown in this paper are significant additions to our previous works shown in \cite{hua2018advanced} and \cite{sohrabi_secrecy}, which expands the understanding of full-duplex radio for wireless network security.

\begin{appendices}
\section{Proof of Lemma 1}
\label{sec:append1}
The following proof is a simple digest from
\cite{tulino2004random} that is useful to help readers to understand Lemma 1 more easily. The Shannon transform of the distribution of a random variable $X$ with parameter $\gamma$ is defined as
\begin{align}
    \mathcal{V}_X (\gamma) = \E_X [\log(1+\gamma X)],
\end{align}
and the $\eta$- transform of the distribution of $X$ with parameter $\gamma$ is defined as
\begin{align}
    \eta_X (\gamma) = \E_X [\frac{1}{1+\gamma X}],
\end{align}
where $\gamma\geq 0$.
The empirical cumulative distribution function of the eigenvalues of an $n \times n$ random non-negative-definite Hermitian matrix $\mathbf{A}$ is defined as
\begin{align}
    F^n_{\mathbf{A}} (x) = \frac{1}{n}\sum_{i=1}^n 1\{\lambda_i(\mathbf{A})\leq x\}
\end{align}
where $\lambda_1(\mathbf{A}), \dots, \lambda_n(\mathbf{A})$ are the eigenvalues of $\mathbf{A}$, and $1\{.\}$ is the indicator function. When $F^n_{\mathbf{A}}(x)$ converges as $n\rightarrow\infty$, the corresponding limit is denoted by $F_{\mathbf{A}}(x)$.

It is obvious that
\begin{align}
       \frac{1}{n} \log |\mathbf{I}+\gamma \mathbf{A}| &= \frac{1}{n} \sum_{i=1}^n \log (1+\gamma \lambda_i (\mathbf{A}))\notag\\& = \int_0^{\infty} \log (1+\gamma x) d F_{\mathbf{A}}^n(x),
\end{align}
and if $n\rightarrow \infty$ then
\begin{align}
       \frac{1}{n} \log |\mathbf{I}+\gamma \mathbf{A}| \rightarrow \int_0^{\infty} \log (1+\gamma x) d F_{\mathbf{A}}(x)
\end{align}
which is the Shannon transform of the eigenvalue distribution  of  the matrix $\mathbf{A}$ when $n$ is large.

The Shannon transform of the eigenvalue distribution of $\mathbf{\Theta}$ with parameter $\eta$ is obviously given by \eqref{eq:shannon}. And the $\eta$-transform of the eigenvalue distribution of $\mathbf{\Theta}$ with parameter $x$ is obviously given by
\begin{align}
\label{eq:eta_theta}
     \eta_{\mathbf{\Theta}}(x)=\frac{1}{L_{\mathbf{\Theta}}}
    \sum_{j=1}^{L_{\mathbf{\Theta}}} \frac{1}{1+x\mathbf{\Theta}_{j,j}}.
\end{align}

From Theorem 2.39 in \cite{tulino2004random}, the $\eta$-transform of the eigenvalue distribution of  $ \mathbf{J}\mathbf{\Theta}\mathbf{J}^H$ with parameter $\gamma$, denoted by $\eta$ here, satisfies
\begin{align}
    \beta = \frac{1-\eta}{1-\eta_{\mathbf{\Theta}}(\gamma\eta)}.
\end{align}
Applying \eqref{eq:eta_theta} to the above equation with $\gamma=1$ yields
\begin{align}
    1-\eta &=\beta \left(1-\frac{1}{L_{\mathbf{\Theta}}}
    \sum_{j=1}^{L_{\mathbf{\Theta}}} \frac{1}{1+\eta\mathbf{\Theta}_{j,j}}\right)
\end{align}
which reduces to \eqref{eq:eq_eta}. Also from Theorem 2.39 in \cite{tulino2004random}, the Shannon transform of the eigenvalue distribution of $\mathbf{J}\mathbf{\Theta}\mathbf{J}^H$ with parameter $\gamma=1$ is
\begin{align}
    \mathcal{V}_{\mathbf{J}\mathbf{\Theta}\mathbf{J}^H}(1)= \beta  \mathcal {V}_{\mathbf{\Theta}}(\eta )-\log \left (\eta\right) +\left (\eta -1\right)\log\left (e\right)
\end{align}
 which is $\Omega(\beta, \mathbf{\Theta}, \eta)$ in  \eqref{eq:omega}.\begin{flushright}\qedsymbol\end{flushright}

\section{Eve uses blind detection with partial knowledge of its CSI}\label{app2}
Now we consider the case where Eve can use its signal in phase 1 to obtain its CSI up to a subspace ambiguity, i.e., in the absence of noise, Eve can obtain from $\mathbf{Y}_E$ as in \eqref{ph1.e} the following:
\begin{equation}\label{}
  \mathbf{\hat A}=\mathbf{A}+\boldsymbol{\Theta}\mathbf{C}_A
\end{equation}
\begin{equation}\label{}
  \mathbf{\hat B}=\mathbf{B}+\boldsymbol{\Theta}\mathbf{C}_B
\end{equation}
where $[\mathbf{C}_A,\mathbf{C}_B]\in \mathbb{C}^{\min\{N_A,N_B\}\times (N_A+N_B)}$ is a known matrix satisfying $[\mathbf{C}_A,\mathbf{C}_B][\mathbf{P}_A^T,\mathbf{P}_B^T]^T=0$. For convenience and without loss of generality, we assume here $a=b=1$.

With one-way information transmission from Alice to Bob in phase 2, Eve can now perform a constrained blind detection as follows:
\begin{equation}\label{}
  \min_{\mathbf{S}\in \mathbb{S}_N^{N_A\times K_2},\mathbf{A}|\mathbf{\hat A}=\mathbf{A}+\boldsymbol{\Theta}\mathbf{C}_A}\|\mathbf{Y}_E-\mathbf{A}\mathbf{S}\|^2
\end{equation}
or equivalently
\begin{equation}\label{eq:P1}
  \min_{\mathbf{S}\in \mathbb{S}_N^{N_A\times K_2},\boldsymbol{\Theta}}\|\mathbf{Y}_E-(\mathbf{\hat A}-\boldsymbol{\Theta}\mathbf{C}_A)\mathbf{S}\|^2.
\end{equation}
For any given $\mathbf{S}$, the solution for $\boldsymbol{\Theta}$ is
\begin{equation}\label{}
  \boldsymbol{\Theta}=-(\mathbf{Y}_E-\mathbf{\hat A}\mathbf{S})(\mathbf{C}_A\mathbf{S})^H
  (\mathbf{C}_A\mathbf{S}(\mathbf{C}_A\mathbf{S})^H)^{-1}.
\end{equation}
Then, the problem of \eqref{eq:P1} reduces to
\begin{equation}\label{eq:P2}
  \min_{\mathbf{S}\in \mathbb{S}_N^{N_A\times K_2}}\|(\mathbf{Y}_E-\mathbf{\hat A}\mathbf{S})(\mathbf{I}_{K_2}-\mathbf{P}_{\mathbf{C}_A\mathbf{S}})\|^2
\end{equation}
where $\mathbf{P}_{\mathbf{C}_A\mathbf{S}}=(\mathbf{C}_A\mathbf{S})^H
  (\mathbf{C}_A\mathbf{S}(\mathbf{C}_A\mathbf{S})^H)^{-1}\mathbf{C}_A\mathbf{S}$. The problem of \eqref{eq:P2} is more complex than \eqref{eq:P0} due to higher order of the cost function in terms of $\mathbf{S}$. A performance analysis of \eqref{eq:P2} can be done in a similar way as for \eqref{eq:P0} but is omitted.

\end{appendices}

\bibliographystyle{IEEEtran}
\bibliography{References}

\end{document}